\documentclass[12pt,a4paper]{report}
\usepackage{amsmath, amssymb, amsthm}
\usepackage{cite}

\usepackage[linesnumbered,lined,ruled,commentsnumbered]{algorithm2e}
\usepackage{graphicx}
\usepackage{url}
\usepackage[dvipsnames]{xcolor}
\usepackage[font=small,skip=4pt]{caption}
\usepackage{chngcntr}
\usepackage{flafter}
\usepackage{enumitem}
\usepackage{parskip}
\usepackage{commath}
\usepackage{optidef}
\usepackage[symbol]{footmisc}
\usepackage{lipsum}
\usepackage{enumitem}
\usepackage{dirtytalk}
\usepackage{subfigure}
\usepackage[subfigure]{tocloft}

\newtheorem{theorem}{Theorem}
\newtheorem{lemma}{Lemma}
\newtheorem{corollary}{Corollary}

\usepackage[top=3cm,bottom=2.5cm,left=3cm,right=3cm]{geometry}

\usepackage{setspace} 

\usepackage{fancyhdr}
\makeindex
\renewcommand{\headrulewidth}{0pt}
\renewcommand*\contentsname{Table of Contents}

\usepackage{enumitem}
\usepackage[
    linktoc=page,
    bookmarks=false,
    colorlinks = true,
    linkcolor = blue,
    urlcolor  = blue,
    citecolor = blue,
    anchorcolor = blue]{hyperref}
\usepackage{cleveref}

\creflabelformat{lemma}{#2\textbf{#1}#3}
\creflabelformat{theorem}{#2\textbf{#1}#3}
\creflabelformat{corollary}{#2\textbf{#1}#3}
\creflabelformat{proposition}{#2\textbf{#1}#3}
\crefname{lemma}{\textbf{Lemma}}{\textbf{Lemmas}}
\crefname{theorem}{\textbf{Theorem}}{\textbf{Theorems}}
\crefname{corollary}{\textbf{Corollary}}{\textbf{Corollaries}}
\crefname{proposition}{\textbf{Proposition}}{\textbf{Propositions}}
\crefname{section}{Section}{Sections}
\crefname{figure}{Fig.}{Figs.}

\makeatletter
\newcommand*{\rom}[1]{\expandafter\@slowromancap\romannumeral #1@}
\makeatother

\begin{document}

\pagenumbering{roman}

\newpage
\begin{titlepage}
    \vspace{\fill} 
    \centering
    \vspace*{1pt}
    \huge{\textbf{Resource Allocation for RIS Assisted CoMP-NOMA Networks using Reinforcement Learning}} \\  [0.4cm]
    \LARGE{Final Year Design Project} \\ [0.7cm]
    \begin{figure}[ht!]
        \centering
        \begin{minipage}{0.45\textwidth}
            \centering
            \includegraphics[width=\linewidth]{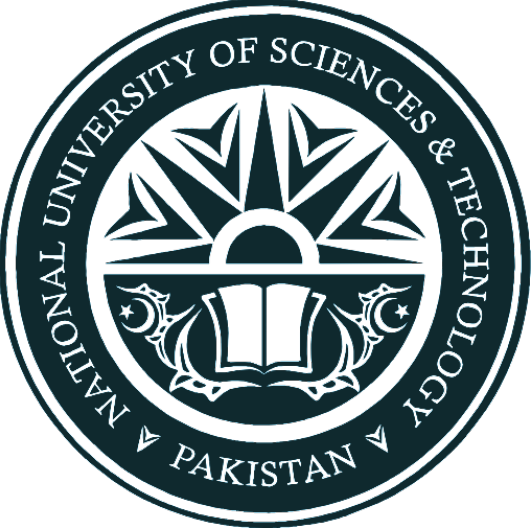}
        \end{minipage}
        \hspace*{2.5em}
        \begin{minipage}{0.37\textwidth}
            \centering
            \includegraphics[width=\linewidth]{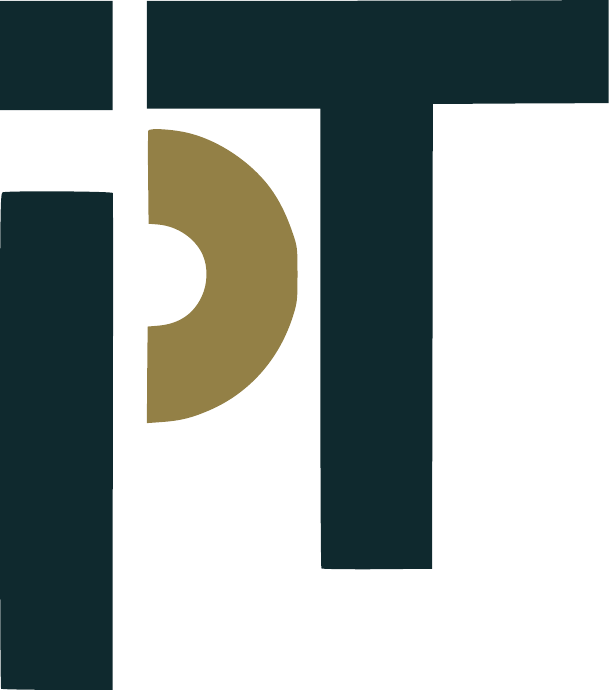}
        \end{minipage}
    \end{figure}
    \vspace {0.7cm}
    \large{By} \\
    \large{\textbf{Muhammad Umer}\quad(CMS -- 345834)} \\
    \large{\textbf{Muhammad Ahmed Mohsin}\quad(CMS -- 333060)} \\ [0.7cm]
    \large{Supervisor: \textbf{Dr. Syed Ali Hassan}} \\
    \large{Co-supervisor: \textbf{Dr. Huma Ghafoor}} \\ [0.7cm]
    \large{School of Electrical Engineering and Computer Science (SEECS) \\
        National University of Sciences and Technology (NUST) \\
        Islamabad, Pakistan} \\ [0.7 cm]
    \large{\today}
    \vspace{\fill} 
\end{titlepage}

\pagenumbering{roman}
\setcounter{page}{1}

\clearpage
\setlength{\headheight}{14pt}

\begin{center}
    {\Huge \bf Certificate}
\end{center}

\vspace{10mm}

It is certified that the contents and form of the thesis entitled

\vspace{5mm}

\begin{center}
    \textbf{\large “Resource Allocation for RIS Assisted CoMP-NOMA Networks using Reinforcement Learning”}
\end{center}

\vspace{5mm}

submitted by \textit{Muhammad Ahmed Mohsin} and \textit{Muhammad Umer} have been found satisfactory for the requirements of the degree.

\vspace{10mm}

\textbf{Advisor:} Syed Ali Hassan, Ph.D.  

Director, \textit{Information Processing and Transmission Lab.}\\Associate Professor, \textit{School of Electrical Engineering and Computer Science,}\\National University of Sciences and Technology (NUST)\\Islamabad, 44000.

\vspace{10mm}

\textbf{Co-Advisor:} Huma Ghafoor, Ph.D.  

Lecturer, \textit{School of Electrical Engineering and Computer Science,}\\  
National University of Sciences and Technology (NUST)\\Islamabad, 44000.

\clearpage
\setlength{\headheight}{14pt}
\begin{center}
    {\Huge \bf Abstract}
\end{center}

This thesis delves into the forefront of wireless communication by exploring the synergistic integration of three transformative technologies: Simultaneously Transmitting and Reflecting Reconfigurable Intelligent Surfaces (STAR-RIS), Coordinated Multi-Point transmission (CoMP), and Non-Orthogonal Multiple Access (NOMA). Driven by the ever-increasing demand for higher data rates, improved spectral efficiency, and expanded coverage in the evolving landscape of 6G development, this research investigates the potential of these technologies to revolutionize future wireless networks.

The thesis analyzes the performance gains achievable through strategic deployment of STAR-RIS, focusing on mitigating inter-cell interference, enhancing signal strength, and extending coverage to cell-edge users. Resource sharing strategies for STAR-RIS elements are explored, optimizing both transmission and reflection functionalities. Analytical frameworks are developed to quantify the benefits of STAR-RIS assisted CoMP-NOMA networks under realistic channel conditions, deriving key performance metrics such as ergodic rates and outage probabilities. Additionally, the research delves into energy-efficient design approaches for CoMP-NOMA networks incorporating RIS, proposing novel RIS configurations and optimization algorithms to achieve a balance between performance and energy consumption. Furthermore, the application of Deep Reinforcement Learning (DRL) techniques for intelligent and adaptive optimization in aerial RIS-assisted CoMP-NOMA networks is explored, aiming to maximize network sum rate while meeting user quality of service requirements. Through a comprehensive investigation of these technologies and their synergistic potential, this thesis contributes valuable insights into the future of wireless communication, paving the way for the development of more efficient, reliable, and sustainable networks capable of meeting the demands of our increasingly connected world.
\clearpage
\setlength{\headheight}{14pt}
\begin{center}
    {\Huge \bf Acknowledgements}
\end{center}

We extend our sincere gratitude to Almighty Allah for granting us the strength, guidance, and perseverance throughout the journey of completing this project. We would like to express our heartfelt appreciation to our parents for their unwavering support, encouragement, and belief in our abilities. Their love, sacrifices, and continuous encouragement have been the cornerstone of our success.

We are deeply grateful to our friends for their camaraderie, motivation, and support, which made this journey memorable and enjoyable. We owe a debt of gratitude to our supervisor, Dr. Syed Ali Hassan, for his invaluable guidance, support, and expertise. His insightful feedback, encouragement, and patience have been instrumental in shaping this thesis.

Special thanks are due to Dr. Haejoon Jung and Dr. Huma Ghafoor for their assistance, valuable insights, and encouragement throughout this project. Their expertise, encouragement, and willingness to help have been truly appreciated. We are indebted to all those who have contributed in any way, no matter how small, to the completion of this project. Thank you for your support, encouragement, and belief in us.
\clearpage
\setlength{\headheight}{14pt}
\begin{center}
    {\Huge \bf Publication List}
\end{center}

The main contributions of this research are either published or accepted or are presently submitted for acceptance in journals and conferences as mentioned in the following list:

\section*{Journal Articles}
\begin{enumerate}[label={[\arabic*]}]
    \item {M. Umer, M. A. Mohsin, M. Gidlund, H. Jung, and S. A. Hassan, “Analysis of STAR-RIS Assisted Downlink CoMP-NOMA Multi-Cell Networks under Nakagami-$m$ Fading,” IEEE
          Communications Letters, 2024.}
\end{enumerate}

\section*{Conference Papers}
\begin{enumerate}[label={[\arabic*]}]
    \item {M. Umer, M. A. Mohsin, S. A. Hassan, H. Jung, and H. Pervaiz, “Performance Analysis of STAR-RIS Enhanced CoMP-NOMA Multi-Cell Networks,” in 2023 IEEE Globecom
          Workshops (GC Wkshps), pp. 2000--2005, IEEE, 2023.}
    \item {M. Umer, M. A. Mohsin, S. A. Hassan, and H. Jung, “On the Energy Efficiency and Passive Beamforming Design of RIS Assisted CoMP-NOMA Networks,” in IEEE ICC 2025.}
    \item {M. Umer, M. A. Mohsin, S. A. Hassan, and H. Jung, “Deep Reinforcement Learning for Trajectory and Beamforming Optimization in CoMP-NOMA Networks with Aerial STAR-RIS,” in IEEE GLOBECOM 2024.}
\end{enumerate}

\clearpage
\renewcommand{\headrulewidth}{0.5pt}
\phantomsection
\tableofcontents
\renewcommand{\contentsname}{Table of Contents}
\addcontentsline{toc}{chapter}{Table of Contents}

\clearpage
\phantomsection
\addcontentsline{toc}{chapter}{List of Figures}
\listoffigures

\clearpage
\phantomsection
\addcontentsline{toc}{chapter}{List of Tables}
\listoftables

\renewcommand\bibname{References}

\clearpage
\phantomsection

\clearpage
\pagenumbering{arabic}
\setcounter{page}{1}

\chapter{Introduction}
\label{chap:intro}
This chapter provides an overview of the research project, outlining its motivation, objectives, methodology, and expected outcomes. It also details the organization of the report for clarity and navigation.

\section{Project Overview}
This research project delves into the synergistic potential of combining three cutting-edge technologies: Simultaneously Transmitting and Reflecting Reconfigurable Intelligent Surfaces (STAR-RIS), Coordinated Multi-Point transmission (CoMP), and Non-Orthogonal Multiple Access (NOMA) to address the challenges of future wireless communication networks, particularly in the context of 6G development. The project investigates the performance enhancements achievable through strategic deployment of STAR-RIS for mitigating inter-cell interference and enhancing signal strength, particularly for cell-edge users. Additionally, it explores resource sharing strategies for STAR-RIS elements to optimize both transmission and reflection functionalities within the network.

The project further analyzes the performance of STAR-RIS assisted CoMP-NOMA networks under realistic channel conditions, deriving analytical expressions for key performance metrics such as ergodic rates and outage probabilities. Moreover, the project investigates energy-efficient design approaches for CoMP-NOMA networks incorporating RIS, proposing different RIS configurations and optimization algorithms for maximizing energy efficiency while maintaining desired performance levels. Finally, the project explores the use of Deep Reinforcement Learning (DRL) techniques, specifically multi-output proximal policy optimization (MO-PPO), to optimize the joint configuration of unmanned aerial vehicles (UAVs), RIS phase shifts, and NOMA power control for maximizing network sum rate and meeting user quality of service requirements.

\section{Motivation}
The ever-increasing demand for higher data rates, improved spectral efficiency, and expanded coverage in wireless communication systems has pushed the boundaries of existing technologies. Traditional approaches are struggling to keep pace with these demands, necessitating the exploration of novel solutions. The limitations of conventional RIS designs, the challenges of interference management in dense heterogeneous networks, and the need for more efficient spectrum utilization motivate the investigation of STAR-RIS, CoMP, and NOMA as potential enablers for next-generation wireless networks. By combining these technologies, we aim to achieve significant performance improvements and pave the way for a more connected and sustainable future.

\section{Objectives}
The primary objectives of this research project are:
\begin{itemize}[]
    \item To analyze and quantify the performance gains achievable through the integration of STAR-RIS in CoMP-NOMA networks, particularly in terms of data rates, coverage, and spectral efficiency.
    \item To develop efficient resource sharing strategies for STAR-RIS elements to optimize both transmission and reflection functionalities within the network.
    \item To design and evaluate energy-efficient solutions for RIS-assisted CoMP-NOMA networks, considering different RIS configurations and optimization algorithms.
    \item To explore the application of DRL techniques for joint optimization of UAV trajectory, RIS phase shifts, and NOMA power control in aerial RIS-assisted CoMP-NOMA networks.
\end{itemize}

\section{Methodology}
The research methodology will involve a combination of theoretical analysis, simulation modeling, and algorithm development. Analytical tools will be employed to derive performance metrics and gain insights into the behavior of STAR-RIS assisted CoMP-NOMA networks under various channel conditions. Simulation models will be developed to evaluate the performance of different system configurations and resource allocation strategies. Optimization algorithms will be designed to maximize energy efficiency, spectral efficiency, and other key performance indicators. Additionally, DRL techniques will be explored for joint optimization in aerial RIS-assisted CoMP-NOMA networks.

\section{Project Outcomes}
The expected outcomes of this research project include:
\begin{itemize}[]
    \item Development of efficient resource sharing strategies for STAR-RIS elements.
    \item A comprehensive understanding of the performance benefits and limitations of STAR-RIS assisted CoMP-NOMA networks.
    \item Design and evaluation of energy-efficient solutions for RIS-assisted CoMP-NOMA networks.
    \item Implementation and analysis of DRL-based optimization techniques for aerial RIS-assisted CoMP-NOMA networks.
    \item Contribution to the advancement of knowledge and understanding of RIS, CoMP, and NOMA technologies for future wireless communication systems.
\end{itemize}

\section{Organization of the Report}
This report is structured to provide a comprehensive and logical flow of information. Chapter~\ref{chap:intro} serves as the introduction, outlining the research project's motivation, objectives, methodology, and expected outcomes. Chapter~\ref{chap:back} delves into the background, offering a detailed review of relevant literature on NOMA, CoMP, RIS, and STAR-RIS technologies. Chapter~\ref{chap:star} explores the synergistic potential of combining STAR-RIS with CoMP-NOMA networks, analyzing their performance under realistic channel conditions. Chapter~\ref{chap:ee} investigates energy-efficient design approaches for RIS-assisted CoMP-NOMA networks. Chapter~\ref{chap:drl} explores the application of DRL techniques for intelligent optimization in aerial RIS-assisted CoMP-NOMA networks. Finally, Chapter 6 concludes the report by summarizing the key findings and discussing potential future research directions.

\chapter{Background}
\label{chap:back}
\section{Preliminaries}
The relentless expansion of wireless communication necessitates continuous innovation to address escalating demands for data rates, spectral efficiency, and coverage. This thesis delves into the exploration of cutting-edge technologies, namely Non-Orthogonal Multiple Access (NOMA), Coordinated Multi-Point (CoMP), and Reconfigurable Intelligent Surfaces (RIS), with a particular focus on Simultaneously Transmitting and Reflecting RIS (STAR-RIS), to tackle the limitations of conventional wireless systems.

This introductory section establishes the foundation for understanding these technologies and their potential impact. It begins by outlining the key challenges faced by modern wireless communication systems and subsequently introduces the aforementioned solutions, highlighting their operating principles and potential benefits.

\subsection{Challenges in Modern Wireless Communication}
The proliferation of data-intensive applications, such as high-definition video streaming, virtual reality, and the Internet of Things (IoT), has placed immense pressure on existing wireless infrastructure. Key challenges include:
\begin{itemize}[]
    \item \textbf{Data Rate Demands:} The need for faster data transfer speeds to support bandwidth-hungry applications is ever-increasing.
    \item \textbf{Spectral Efficiency:} Efficient utilization of the limited radio spectrum is crucial to accommodate the growing number of users and devices.
    \item \textbf{Coverage:} Ensuring ubiquitous and reliable wireless coverage across diverse environments, including urban canyons and remote areas, remains a significant challenge.
\end{itemize}

Traditional approaches, such as increasing carrier frequencies or deploying additional base stations, offer limited scalability and often introduce new complexities. Consequently, exploring and developing novel technologies is essential to overcome these hurdles and pave the way for next-generation wireless communication systems.

\subsection{Emerging Technologies}
Several promising technologies have emerged to address the limitations of conventional wireless systems:

\subsubsection{Non-Orthogonal Multiple Access (NOMA)}
NOMA departs from the traditional orthogonal multiple access (OMA) paradigm by leveraging power-domain multiplexing to serve multiple users on the same resource block. By allocating higher power levels to users experiencing weaker channel conditions, NOMA ensures successful decoding at the receiver while simultaneously improving spectral efficiency and user fairness. This technique offers significant advantages over OMA, including increased system capacity and enhanced service for cell-edge users.

\subsubsection{Coordinated Multi-Point (CoMP)}
CoMP enhances network performance through coordinated transmission and reception among multiple base stations. This coordination effectively mitigates inter-cell interference (ICI), leading to improved signal quality, increased network capacity, and extended coverage, particularly at cell edges. However, implementing CoMP presents challenges related to channel state information (CSI) acquisition and synchronization among base stations.

\subsubsection{Reconfigurable Intelligent Surfaces (RIS)}
RIS introduces a novel paradigm by employing software-controlled metasurfaces to manipulate electromagnetic waves. These surfaces, comprised of numerous reconfigurable elements, can dynamically adjust the phase and amplitude of incident waves to achieve desired reflection, refraction, or scattering. While offering significant potential for enhancing wireless environments, traditional RIS designs suffer from the limitation of only operating in the reflection or transmission mode, known as the \say{half-space} problem.

\subsubsection{Simultaneously Transmitting and Reflecting RIS (STAR-RIS)}
STAR-RIS overcomes the limitations of traditional RIS by incorporating active elements alongside passive reflectors. This enables simultaneous transmission and reflection of signals, providing greater flexibility in manipulating electromagnetic waves and further enhancing wireless communication capabilities. STAR-RIS holds immense potential for improving signal coverage, boosting system capacity, and increasing spectral efficiency, particularly in challenging propagation environments.

The subsequent sections of this thesis will delve deeper into these technologies, exploring their theoretical foundations, practical implementations, and potential applications in the context of future wireless communication systems.

\section{Related Work}
This section explores existing research on key technologies relevant to this thesis: Non-Orthogonal Multiple Access (NOMA), Coordinated Multi-Point (CoMP), Reconfigurable Intelligent Surfaces (RIS), and their application in multi-cell networks. These technologies address critical challenges in modern wireless communication systems, aiming to increase data rates, improve spectral efficiency, and enhance coverage.

\subsection{Non-Orthogonal Multiple Access (NOMA)}
The work in~\cite{liu2018nonorthogonal} delves into the potential of NOMA for 5G and beyond networks. Recognizing the limitations of traditional orthogonal multiple access (OMA) in accommodating the escalating demands of data-intensive applications, the authors highlight the advantages of NOMA in enhancing bandwidth efficiency by serving multiple users within the same resource block. The paper provides a comprehensive overview of power-domain multiplexing aided NOMA, encompassing theoretical foundations, multiple antenna designs, interaction with cooperative transmission schemes, resource management strategies, and co-existence with other 5G technologies. Additionally, it identifies existing challenges and proposes potential solutions, offering valuable design guidelines and outlining future research directions in this domain.

Further exploring NOMA, the survey in~\cite{Islam_2017} offers a detailed analysis of its progress in 5G systems, focusing primarily on power-domain NOMA employing superposition coding at the transmitter and successive interference cancellation at the receiver. The paper delves into various aspects of NOMA, including capacity analysis, power allocation strategies, user fairness considerations, and integration with other wireless technologies such as cooperative communications, MIMO, beamforming, and network coding. Implementation challenges and potential avenues for future research are also discussed, providing a holistic understanding of NOMA and its potential impact on future wireless networks.

\subsection{Coordinated Multi-Point (CoMP)}
The study presented in~\cite{lee2012coordinated} investigates CoMP techniques for LTE-Advanced systems, focusing on coordinated transmission and reception among multiple points to mitigate interference and enhance signal quality. The authors evaluate the potential performance benefits of CoMP across various deployment scenarios with varying traffic loads. Furthermore, the study delves into implementation aspects, practical challenges, and deployment considerations for CoMP in LTE-Advanced networks.
Building upon the concept of CoMP, the work in~\cite{nigam2014coordinated} explores base station cooperation in the downlink of heterogeneous cellular networks, specifically focusing on joint transmission scenarios. Using stochastic geometry, the authors derive expressions for network coverage probability, considering a typical user receiving data from a pool of cooperating base stations. The analysis demonstrates significant gains in coverage probability, particularly when cooperation involves multiple base stations, and emphasizes the importance of coherent joint transmission in achieving diversity gain.

\subsection{Reconfigurable Intelligent Surfaces (RIS)}
The study in~\cite{huang2019reconfigurable} investigates the use of RIS for enhancing downlink multi-user communication in a multi-antenna base station scenario. The authors focus on developing energy-efficient resource allocation strategies for transmit power and RIS phase shifts while guaranteeing individual link budget requirements for users. To address the resulting non-convex optimization problems, two computationally efficient algorithms are proposed. The paper also introduces a realistic power consumption model for RIS-based systems and evaluates the proposed methods in a realistic outdoor environment. Results demonstrate significant energy efficiency gains compared to conventional relaying techniques.

Further exploring RIS technology, the work in~\cite{wu2019towards} provides a comprehensive overview of its potential for revolutionizing wireless communication networks. The paper delves into the applications, advantages, hardware architecture, and signal model of RIS, highlighting its ability to intelligently manipulate the wireless propagation environment using passive reflecting elements. The authors also discuss the challenges associated with designing and implementing hybrid wireless networks incorporating both active and passive components. Numerical results showcase the performance improvements achievable with RIS in typical wireless network scenarios.

\subsection{UAV-assisted Wireless Networks}
While not directly related to RIS, CoMP, or NOMA, the inclusion of UAV-assisted wireless networks in the original literature review presents an opportunity to discuss the potential integration of these technologies. The survey in~\cite{zhang2019survey} discusses the increasing role of unmanned aerial vehicles (UAVs) in enhancing transmission efficiency and coverage in wireless communication systems, particularly in the context of 5G and beyond networks utilizing millimeter wave (mmWave) frequencies. The paper offers a comprehensive overview of integrating 5G mmWave communications into UAV-assisted networks, presenting a taxonomy of research issues and solutions, and discussing technical advantages, challenges, and potential applications. This opens avenues for exploring the synergistic benefits of combining UAV-assisted networks with RIS, CoMP, and NOMA to further enhance coverage, capacity, and spectral efficiency in future wireless systems.

Furthermore, the work in~\cite{liu2019trajectory} proposes a novel framework for designing trajectories of multiple UAVs by predicting users' mobility information. This framework aims to maximize the instantaneous sum transmit rate while satisfying user rate requirements. Integrating such trajectory design and user mobility prediction techniques with RIS, CoMP, and NOMA could pave the way for highly efficient and adaptable wireless networks capable of dynamically responding to user demands and environmental changes.

\section{Related Work}
This section explores existing research on the integration of STAR-RIS, CoMP, and NOMA technologies in various wireless communication scenarios. The focus is on understanding how these technologies are combined to address challenges related to coverage, spectral efficiency, and energy efficiency.

\subsection{STAR-RIS Enhanced CoMP-NOMA Networks}
The work in~\cite{hou2021joint} introduces a novel design for STAR-RIS within a NOMA-enhanced CoMP network. Building upon existing signal enhancement and cancellation approaches, the authors propose a simultaneous signal enhancement and cancellation based (SSECB) design that leverages a large number of RIS elements to concurrently eliminate inter-cell interference and boost desired signals. Simulation results demonstrate the effectiveness of SSECB in outperforming conventional designs and achieving superior performance in CoMP-NOMA networks.

Further exploring the potential of STAR-RIS in NOMA systems, the study in~\cite{xu2022secrecy} investigates its application in enhancing coverage quality and spectral efficiency. The authors analyze the secrecy performance of a STAR-RIS aided downlink NOMA system employing the energy splitting protocol. Analytical expressions for secrecy outage probability (SOP) are derived, and asymptotic performance analysis is conducted to gain insights into system behavior. The results demonstrate the superior secrecy performance of STAR-RIS-NOMA compared to conventional OMA systems.

\subsection[DRL for Aerial RIS in CoMP-NOMA Networks]{Deep Reinforcement Learning (DRL) for Aerial RIS in CoMP-NOMA Networks}

The integration of UAVs with RIS in CoMP-NOMA networks presents exciting opportunities for further enhancing coverage and capacity. The work in~\cite{chen2024comp} proposes a novel approach for maximizing communication efficiency in a multi-UAV system assisted by RIS. A hybrid learning scheme combining multi-agent DRL and alternating optimization is employed to optimize UAV trajectories, cooperative beamforming, and RIS passive beamforming (PBF). The proposed framework demonstrates superior performance compared to conventional systems and heuristic algorithms, achieving higher communication rates and fast convergence.

Further exploring the use of RIS in UAV-assisted communication, the study in~\cite{mei20223d} investigates the joint optimization of UAV placement and RIS phase-shift to maximize data transfer rates while minimizing UAV energy consumption. DRL algorithms, specifically DDQN and DDPG, are employed to address the optimization challenges. Numerical results demonstrate the effectiveness of these algorithms in improving the energy efficiency of RIS-assisted UAV systems compared to benchmark solutions.

The work in~\cite{10051712} introduces an energy harvesting (EH) scheme for UAV-RIS systems operating in communication-disabled areas. The proposed EH-RIS scheme utilizes SWIPT technology to simultaneously transport information and harvest energy through a split passive reflecting array. A robust DRL-based algorithm is developed for efficient resource allocation in dynamic environments with pedestrian mobility and rapid channel changes. Simulation results showcase the effectiveness and efficiency of the proposed EH-RIS system, surpassing existing solutions and achieving near-optimal performance.

\subsection[Energy Efficiency]{Energy Efficiency of RIS-assisted Multi-Cell NOMA Networks}

Energy efficiency (EE) is a critical consideration in the design of future wireless networks. The paper in~\cite{adam2020energy} focuses on maximizing EE in multi-cell multi-carrier NOMA networks while considering hardware impairments. The authors propose a two-stage approach, employing the BWOA algorithm for user association and subchannel assignment, and SPCA for power allocation. Simulation results demonstrate the effectiveness of the proposed algorithm in achieving comparable performance to existing methods while surpassing benchmarks for both NOMA and OMA systems.

Another study on EE maximization in multi-cell multi-carrier NOMA networks is presented in~\cite{adam2020energy}. The authors propose a matching-based framework for user association and a two-stage quadratic transform approach for power allocation to address the non-convex EE maximization problem. Numerical results demonstrate the superior EE performance of the proposed method compared to existing approaches for NOMA and OMA systems.

Expanding on the concept of EE, the work in~\cite{long2020joint} explores secure energy efficiency maximization in an RIS-assisted uplink wireless communication system involving a UAV acting as a mobile relay. The authors propose an algorithm for joint optimization of UAV trajectory, RIS phase shift, user association, and transmit power to maximize system efficiency. Simulation results demonstrate significant gains in secure energy efficiency compared to traditional schemes without RIS integration.

\section{Research Gaps and Opportunities}
This section identifies key research gaps and opportunities in the existing literature on RIS-enhanced CoMP-NOMA networks. Addressing these gaps is crucial for realizing the full potential of these technologies in future wireless communication systems.
\subsection{Performance Challenges in CoMP-NOMA Networks}
While CoMP-NOMA networks offer significant potential for enhancing spectral efficiency, several challenges hinder their practical implementation:

\begin{itemize}[]
    \item \textbf{Channel State Information (CSI) Acquisition:} Accurate CSI is crucial for effective coordination among base stations in CoMP systems. However, acquiring accurate CSI in dynamic wireless environments with RIS deployments is a complex task due to the additional channel paths introduced by the RIS. Developing efficient and robust CSI acquisition techniques tailored for RIS-assisted CoMP-NOMA networks is essential.
    \item \textbf{Synchronization Issues:} Precise synchronization among base stations is critical for CoMP operation to avoid inter-cell interference. The introduction of RIS further complicates synchronization requirements due to the potential for additional delays and phase shifts introduced by the reflecting elements. Developing robust synchronization protocols that account for the unique characteristics of RIS is necessary.
    \item \textbf{Signal Processing Complexity:} NOMA requires more complex signal processing compared to traditional OMA techniques, both at the base stations and user equipment. This increased complexity can lead to higher power consumption and computational costs. Efficient signal processing algorithms optimized for RIS-assisted CoMP-NOMA systems are needed to mitigate these challenges.

\end{itemize}

\subsection{Open Research Areas in STAR-RIS Technology}
STAR-RIS presents a promising evolution of RIS technology, but its integration into wireless communication systems requires further investigation in several key areas:

Channel Modeling: Accurate channel models that capture the unique characteristics of STAR-RIS, including its simultaneous transmitting and reflecting capabilities, are crucial for system design and performance evaluation. These models should account for mutual coupling between elements, the interaction between transmitted and reflected signals, and the impact of different RIS configurations on channel characteristics.

\begin{itemize}[]
    \item \textbf{Capacity Analysis:} Determining the theoretical limits of STAR-RIS systems in terms of achievable data rates is essential for understanding its potential benefits and limitations. This analysis should consider various factors such as RIS size, element spacing, operating frequency, and channel conditions.
    \item \textbf{Beamforming Optimization:} Designing efficient beamforming algorithms for STAR-RIS that jointly optimize reflection and transmission while considering practical constraints is a challenging task. Advanced optimization techniques are needed to achieve optimal performance while accounting for factors like power limitations, hardware imperfections, and channel uncertainties.
    \item \textbf{Performance Comparison:} A comprehensive comparison of STAR-RIS with existing technologies such as conventional RIS and relaying systems under various scenarios is crucial for evaluating its advantages and identifying suitable application areas. This comparison should consider metrics like energy efficiency, spectral efficiency, coverage, and cost.
\end{itemize}

\subsection[Statistical Gaps]{Statistical Analysis Gaps in RIS-based CoMP-NOMA Networks}

While existing research has explored the integration of STAR-RIS into CoMP-NOMA networks, there is a lack of a comprehensive statistical framework for analyzing and optimizing their performance. Key gaps include:
\begin{itemize}[]
    \item \textbf{Performance Analysis and Optimization:} Existing studies often rely on simplified assumptions and specific scenarios when analyzing outage probability and rate coverage. More comprehensive analyses incorporating diverse user distributions, varying channel conditions, and realistic RIS configurations are needed. Additionally, developing statistically robust algorithms for optimal user pairing, power allocation, and joint beamforming/phase shift design is crucial.
    \item \textbf{Resource Management and Interference Mitigation:} Dynamically adapting RIS configuration to varying channel conditions and user demands requires intelligent resource management strategies. Statistical methods can provide valuable tools for predicting channel variations and optimizing RIS phase shifts accordingly. Furthermore, statistical techniques can aid in developing efficient interference cancellation and mitigation schemes tailored for RIS-assisted CoMP-NOMA networks.
    \item \textbf{Security and Privacy:} The broadcast nature of RIS reflections introduces security and privacy concerns. Statistical approaches can be leveraged to design secure transmission protocols and privacy-preserving mechanisms that mitigate potential vulnerabilities in RIS-assisted CoMP-NOMA networks.
\end{itemize}

\subsection[DRL in Aerial RIS Networks]{DRL in Aerial RIS Networks: Challenges and Opportunities}
While DRL has shown promising results in optimizing aerial RIS networks, several research gaps and challenges need to be addressed for practical real-world deployment:

\subsubsection{Expanding the Scope of Scenarios and Environments}
\begin{itemize}[]
    \item \textbf{Dynamic Environments:} Most existing research focuses on static environments, neglecting the dynamic nature of real-world scenarios with moving objects and changing channel conditions. Investigating DRL techniques in such dynamic environments is crucial for practical applications.
    \item \textbf{Complex Network Topologies:} Current studies primarily consider simple network topologies with limited aerial platforms and RIS elements. Exploring DRL algorithms in more complex networks with multiple aerial layers, diverse user distributions, and heterogeneous platforms (e.g., drones and high-altitude platforms) is essential for expanding the applicability of these techniques.
\end{itemize}

\subsubsection{Enhancing Algorithm Efficiency and Scalability}
\begin{itemize}[]
    \item \textbf{Computational Complexity:} Many DRL algorithms suffer from high computational complexity and long training times, hindering real-time implementation. Research into efficient DRL algorithms with reduced complexity and faster convergence is needed to overcome this limitation.
    \item \textbf{Scalability:} The performance of DRL algorithms often degrades as network size and complexity increase. Developing scalable DRL solutions capable of efficiently handling large-scale aerial RIS networks is vital for practical deployment.
\end{itemize}

\subsubsection{Joint Optimization and Hybrid Approaches}
\begin{itemize}[]
    \item \textbf{Integration with Path Planning:} Investigating the joint optimization of aerial platform trajectories and RIS configurations using DRL has the potential to significantly improve overall network performance.
    \item \textbf{Hybrid DRL Approaches:} Combining DRL with other optimization techniques, such as evolutionary algorithms or convex optimization methods, could leverage the strengths of each approach and lead to improved performance and convergence speed.
\end{itemize}

\subsubsection{Security and Privacy Considerations}
\begin{itemize}[]
    \item \textbf{Adversarial Attacks:} The vulnerability of DRL algorithms to adversarial attacks is a growing concern. Research into robust DRL techniques that are resilient to malicious attacks is necessary for secure and reliable operation of aerial RIS networks.
    \item \textbf{Privacy Preservation:} DRL algorithms may require access to sensitive user data, raising privacy concerns. Investigating privacy-preserving DRL techniques that protect user information while maintaining performance is critical for ethical and responsible implementation of these technologies.
\end{itemize}

\subsection[Energy Efficiency Assessment in Aerial RIS Networks]{Energy Efficiency Assessment in Aerial RIS Networks: A Holistic Approach}
Accurately assessing and optimizing energy efficiency in aerial RIS networks requires a holistic approach that considers various factors and challenges:

\subsubsection{Developing Comprehensive Energy Consumption Models}
\begin{itemize}[]
    \item \textbf{Holistic Network Perspective:} Existing models often focus on individual components, neglecting the interconnected nature of the entire network. Comprehensive models considering UAVs, RIS, communication links, and ground infrastructure are needed for accurate energy efficiency assessments.
    \item \textbf{Dynamic Energy Consumption:} Models should account for real-time factors such as UAV trajectory optimization, varying channel conditions, and RIS configuration changes to accurately capture the dynamic energy consumption of aerial RIS networks.
\end{itemize}

\subsubsection{Optimization Algorithms with Energy Efficiency Focus}
\begin{itemize}[]
    \item \textbf{Balancing Performance and Energy Efficiency:} Algorithms should be designed to optimize network performance with energy efficiency as a primary or secondary objective, effectively balancing these often conflicting goals.
    \item \textbf{Joint Optimization:} Exploring joint optimization of UAV trajectory, RIS phase shifts, and communication resource allocation with an emphasis on energy efficiency holds significant potential for improving overall network sustainability.
\end{itemize}

\subsubsection{Addressing Practical Considerations and Implementation Challenges}
\begin{itemize}[]
    \item \textbf{Hardware Limitations:} Realistic models and optimization algorithms need to account for the limitations of RIS elements, UAVs, and communication systems to ensure practical feasibility and avoid overestimating potential gains.
    \item \textbf{Environmental Factors:} Environmental conditions, such as weather and wind speed, can significantly impact energy consumption. These factors should be considered in energy efficiency assessments and optimization strategies.
    \item \textbf{Cost Analysis:} Evaluating the economic feasibility of aerial RIS networks through cost-benefit analysis and comparison with terrestrial networks is crucial for determining their viability and potential for large-scale deployment.
\end{itemize}

\subsubsection{Standardization and Protocols}
\begin{itemize}[]
    \item \textbf{Developing Standardized Protocols:} As a relatively new technology, aerial RIS networks lack standardized protocols for energy management and communication. Establishing such protocols is essential for ensuring interoperability and facilitating large-scale deployment.
    \item \textbf{Energy-efficient Communication Protocols:} Existing communication protocols may not be optimized for energy efficiency in aerial RIS networks. Research on adapting existing protocols or developing new energy-aware protocols is necessary to minimize energy consumption while maintaining reliable communication.

\end{itemize}

\subsubsection{Security and Privacy in Energy-Efficient Systems}
\begin{itemize}[]
    \item \textbf{Energy-efficient Security Mechanisms:} Implementing security mechanisms can increase energy consumption. Research is needed on designing energy-efficient security solutions for aerial RIS networks that balance security requirements with energy efficiency goals.
    \item \textbf{Privacy Concerns:} The dynamic nature of aerial RIS networks raises privacy concerns. Studies should explore privacy-preserving mechanisms that protect user information while maintaining energy efficiency and network performance.

\end{itemize}

\section{Summary}
This chapter has laid the foundation for exploring the potential of cutting-edge technologies in addressing the ever-growing demands of modern wireless communication. The escalating need for higher data rates, improved spectral efficiency, and expanded coverage has exposed the limitations of traditional approaches, necessitating the investigation of novel solutions.

Non-Orthogonal Multiple Access (NOMA) was introduced as a promising candidate for enhancing spectral efficiency by enabling multiple users to share the same resource block, offering advantages over conventional Orthogonal Multiple Access (OMA) schemes. Additionally, Coordinated Multi-Point (CoMP) techniques were presented as a means to mitigate inter-cell interference and improve network performance through coordinated transmission and reception among base stations. However, challenges related to channel state information acquisition, synchronization, and signal processing complexity remain hurdles for the practical implementation of both NOMA and CoMP.

Reconfigurable Intelligent Surfaces (RIS) emerged as a revolutionary technology capable of manipulating electromagnetic waves through software-controlled metasurfaces, offering a new paradigm for shaping wireless environments. The limitations of traditional RIS designs, particularly the \say{half-space} problem, were addressed with the introduction of Simultaneously Transmitting and Reflecting RIS (STAR-RIS) as a more versatile and powerful solution.

By establishing a firm understanding of the fundamental principles and limitations of NOMA, CoMP, RIS, and STAR-RIS, this chapter has paved the way for further exploration and analysis of their potential in shaping the future of wireless communication systems. Subsequent chapters will delve deeper into their applications and evaluate their impact on the evolution of wireless networks, striving to overcome existing challenges and unlock their full potential.

\chapter[Synergy of STAR-RIS, CoMP, and NOMA]{Synergy of STAR-RIS, CoMP, and NOMA: Paving the Path for 6G Wireless Networks}
\label{chap:star}
The evolution towards sixth-generation (6G) wireless networks demands significant advancements in spectral efficiency and coverage to accommodate the ever-growing demand for data-driven applications and ubiquitous connectivity. Reconfigurable Intelligent Surfaces (RIS) have emerged as a transformative technology with the potential to revolutionize wireless communication by intelligently manipulating electromagnetic waves. RIS can dynamically shape the propagation environment to enhance signal strength and mitigate interference, thus, facilitating higher data rates and improved spectral efficiency.

In parallel, the increasing deployment of small, low-power base stations within cellular networks has led to challenges related to cross-tier interference and increased energy consumption. Coordinated Multi-Point (CoMP) techniques offer a solution by enabling base stations to coordinate transmissions through high-speed fronthaul links, thereby mitigating interference and enhancing overall network performance. Integrating CoMP with Non-Orthogonal Multiple Access (NOMA) further improves spectral efficiency by allowing multiple users to share the same resource block.

This chapter explores the synergistic potential of combining STAR-RIS with CoMP-NOMA networks to address the challenges of future wireless communication systems. By leveraging the unique capabilities of each technology, we aim to pave the path towards achieving the ambitious goals of 6G networks, providing enhanced coverage, improved spectral efficiency, and a more sustainable and connected future.

\LARGE{\textbf{Performance Analysis}}
\normalsize

\section{System Model}

As shown in Figure~\ref{fig:ana_system}, we consider a multi-cell STAR-RIS assisted CoMP-NOMA network, in which each base station (BS) serves a NOMA pair consisting of its corresponding center and edge users. Consequently, the edge user is part of two NOMA pairs, each served by a different BS.

\begin{figure}[h!]
    \centering
    \includegraphics[width=0.65\columnwidth]{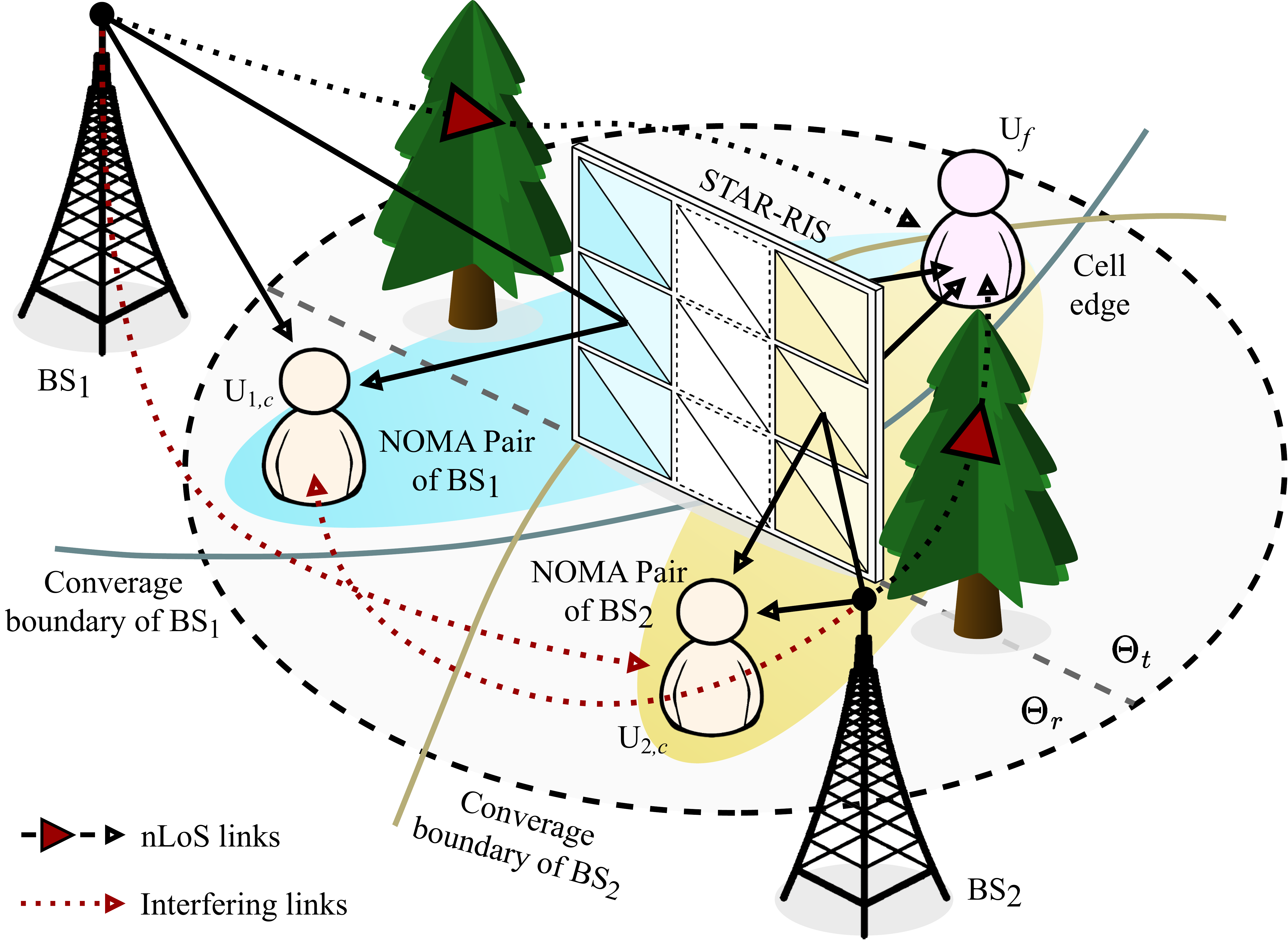}
    \vspace*{1em}
    \caption{An illustration of STAR-RIS-aided coordinated NOMA cluster.}
    \label{fig:ana_system}
\end{figure}

We define the index set $\mathcal{I} = \{1, 2\}$ for the two BSs, whereas  $\mathcal{C} = \{1, 2,\dots, C\}$ is defined for the cell-center users, and $\mathcal{F} = \{1, 2,\dots, F\}$ for the cell-edge users, respectively. Additionally, let $\mathcal{U} = \mathcal{C}\cup\mathcal{F}$, denote the set of all system users. For the sake of simplicity, we assume $C=2$, where each center user U$_{c_i}$ is served by its corresponding BS$_i$, and $F=1$, with $C$ and $F$ representing the cardinality of $\mathcal{C}$ and $\mathcal{F}$, respectively.

The BSs employ power-domain NOMA techniques to communicate with the users. Specifically, $\forall i \in \mathcal{I}, c \in \mathcal{C}$, and $f \in \mathcal{F}$, BS$_i$ forms the NOMA pair (${\text{U}_{c_i}, \text{U}_f}$), where $\text{U}_{c_i}$ is the cell-center user of BS$_i$. Consequently, the U$_f$ is part of two NOMA pairs, each cluster served by a different BS. To mitigate the strong ICI experienced by U$_f$, CoMP is adopted between the two BSs.
In addition, it is assumed that the BSs are connected to a central processing unit (CPU) via high-speed fronthaul links, facilitating seamless information sharing and coordinated transmissions among them.

In this work, perfect CSI is assumed to be available at the BSs. While this is a challenging assumption in practice, recent advances in channel estimation techniques for RIS-enabled wireless networks have shown that it is possible to achieve accurate CSI~\cite{hou2021joint, wei2021channel, taha2021enabling, 10462124, 10847914} with a reasonable amount of overhead.

\subsection{Channel Model}
For each communication link in the system, we take into account both large-scale fading and small-scale fading effects. Due to the relatively large propagation distances and the scattering effect of the links between B$_i$ and U$_u$, $\forall i \in \mathcal{I}$ and $u \in \mathcal{U}$, the channels are assumed to follow Rayleigh fading, expressed as:
\begin{equation}
    h_{i,u} = \sqrt{\frac{\rho_{o}}{PL(d_{i,u})}} v_{i,u},
\end{equation}
where $v_{i,u}$ is a complex Gaussian random variable that follows a Rayleigh distribution with zero mean and unit variance, $\rho_{o}$ is the reference path-loss at a distance of 1 m, $PL(d_{i,u})$ is the large scale path-loss, modeled as $PL(d_{i,u})=\left(d_{i,u}\right)^{\alpha_{i\rightarrow u}}$, where $d_{i, u}$ is the distance and $\alpha_{i\rightarrow u}$ is the path-loss exponent between the BS$_i$ and U$_u$, respectively.

On the contrary, the link between the STAR-RIS, hereafter represented by $R$, and BS$_i$ is assumed to exhibit a dominant line-of-sight (LoS) path~\cite{guo2020intelligent}. Therefore, these links are subject to the Rician fading, where their channel coefficients are expressed as:
\begin{equation}
    \textbf{h}_{i,R} = \sqrt{\frac{\rho_{o}}{PL(d_{i,R})}} \left( \sqrt{\frac{\kappa_{i,R}}{\kappa_{i,R} + 1}} \hat{\mathbf{v}_{i,R}} + \sqrt{\frac{1}{\kappa_{i,R} + 1}} \mathbf{v}_{i,R} \right),
\end{equation}
where $d_{i,R}$ is the distance between the BS$_i$ and $R$, $\kappa_{i,R}$ represents the Rician factor, $\hat{\mathbf{v}_{i,R}}$ represents the deterministic LoS components, and $\mathbf{v}_{i,R}$ denotes the complex Gaussian random variables, each following a Rayleigh distribution with zero mean and unit variance, thus representing the NLoS components. The links between $R$ and U$_u$, $\forall u \in \mathcal{U}$, are also modeled similarly.

\subsection{STAR-RIS Configuration}
The energy splitting (ES) model of the STAR-RIS array can be mathematically characterized by the following respective transmission- and reflection-coefficient matrices~\cite{mu2021simultaneously}:
\begin{align}
    \mathbf{\Theta_r} & = \sqrt{\beta^r}\text{diag}(e^{j \theta_1^t}, e^{j \theta_2^t}, \dots, e^{j \theta_K^t}), \\
    \mathbf{\Theta_t} & = \sqrt{\beta^t}\text{diag}(e^{j \theta_1^r}, e^{j \theta_2^r}, \dots, e^{j \theta_K^r}),
\end{align}
where $\beta^t$, $\beta^r \in [0, 1]$ and $\theta_k^t$, $\theta_k^r \in [0, 2 \pi)$, $\forall k \in \mathcal{K} \triangleq \{1, 2,\dots,K\}$. The phase shifts for transmission and reflection (i.e., $\theta_k^t$ and $\theta_k^r$) can generally be chosen independently of each other~\cite{9437234}. However, the amplitude adjustments for transmission and reflection are coupled by the law of conservation of energy. Assuming the STAR-RIS does not impose any power loss, the relation between the amplitude coefficients (i.e., $\beta^t$ and $\beta^r$) is expressed as $\beta^t + \beta^r = 1$. To reduce the signaling overhead between the STAR-RIS and the BSs, all elements are adjusted to have the same transmission and reflection coefficients.

\subsection{Rate Analysis}
To analyze the rates achieved for users in the system model shown in Figure~\ref{fig:ana_system}, we first present the signal model. Specifically, $\forall i \in \mathcal{I}$, $c \in \mathcal{C}$, and $f \in \mathcal{F}$, let the tuple (U$_{c_1}$, U$_{c_2}$, U$_f$) represent the coordinated NOMA cluster, and let $P_1$ and $P_2$ denote the transmit powers of BS$_1$ and BS$_2$, respectively. This signal model serves as the foundation for evaluating the achieved rates and optimizing the system performance in the considered cluster. Additionally, let $x_{1,c}$, $x_{2,c}$, and $x_f$ represent the message signal intended for U$_{c_1}$, U$_{c_2}$, and U$_f$, respectively. Each BS$_i$ broadcasts a superimposed signal of the messages intended for users within its coverage region, U$_{c_i}$ and U$_f$, and expressed as~\cite{saito2013non}:
\begin{equation}
    x_{i}=\sqrt{\zeta_{i,c}P_i}x_{i,c} + \sqrt{\zeta_{i,f}P_i}x_f,
\end{equation}
where $\zeta_{i,c}$ and $\zeta_{i,f}$ are the power allocation (PA) factors assigned by BS$_i$ to users U$_{c_i}$ sand U$_f$, respectively. It is important to note that U$_{c_i}$ experiences stronger channel conditions compared to U$_f$, making it the dominant NOMA user in the pair ($\text{U}_{c_i}, \text{U}_f$) formed by BS$_i$. Following the principle of NOMA, U$_{c_i}$ should be capable of detecting and decoding the message intended for U$_f$. This principle also implies that $\zeta_{i,c} < 0.5$, or $0.5 < \zeta_{i,f} < 1$~\cite{obeed2020user, salem2020noma}.

For brevity, we only define the rate achieved by U$_{c_1}$ from the set of cell-center users $\mathcal{C}$, as the same steps could be extended to define the rate of U$_{c_i}$, $\forall i \in \mathcal{I}$ and $c \in \mathcal{C}$. The received signal at U$_{c_1}$ can be written as:
\begin{equation}
    y_{c_1}=h_{1,c}x_1 + h_{2,c^\prime} x_2 + N_o,
\end{equation}
where $N_o$ is an additive white Gaussian noise (AWGN), i.e., $N_o\sim \mathcal{CN}$(0, $\sigma^2$). Further, $h_{2,c^\prime}$ is the channel corresponding to the link between BS$_2$ and U$_{c_1}$, which is the cell-center user of BS$_1$, and represents the ICI experienced by U$_{c_1}$. By utilizing successive interference cancellation (SIC) techniques, U$_{c_1}$ first decodes the message signal of U$_f$ (i.e., $x_f$) and then removes it from $y_{c_1}$ to decode its own message (i.e., $x_{1,c}$). Based on this approach, the signal-to-interference-and-noise ratio (SINR) and the corresponding achievable rate at U$_{c_1}$ for decoding the message of U$_f$ can be expressed as:
\begin{gather}
    \label{eq:gamma_icf}
    \gamma_{1,c\rightarrow f}=\frac{\zeta_{1,f}P_1\abs{H_{1,c}}^2}{\zeta_{c_1}P_1\abs{H_{1,c}}^2 + P_2\abs{h_{2,c^\prime}}^2 +  \sigma^2}, \\
    \mathcal{R}_{1,c\rightarrow f}=\log_2\left(1+\gamma_{1,c\rightarrow f}\right),
\end{gather}
where $H_{1,c}=h_{1,c}+\textbf{h}_{R, c}^H \mathbf{\Theta_r}\textbf{h}_{1, R}$ represents the combined channel from BS$_1$ to U$_{c_1}$. Furthermore, the SINR and the corresponding achievable rate of U$_{c_1}$ for decoding its own message can be expressed as:
\begin{gather}
    \gamma_{c_1}=\zeta_{c_1}\frac{P_1\abs{H_{1,c}}^2}{P_2\abs{h_{2,c^\prime}}^2 + \sigma^2}, \\
    \mathcal{R}_{c_1} = \log_2\left(1+\gamma_{c_1}\right).
\end{gather}

On the contrary, U$_f$, belonging to two NOMA pairs, receives its signal through the broadcasts from each BS$_i$, $\forall i \in \mathcal{I}$. Thus, the received signal at U$_f$ can be expressed as:
\begin{equation}
    y_f = H_{1, f}x_1 + H_{2, f}x_2 + N_0,
\end{equation}
where $H_{1, f}$ and $H_{2, f}$ represent the combined channels from BS$_1$ to U$_f$ and from BS$_2$ to U$_f$, and can be expressed as $H_{1, f}=h_{1, f}+\textbf{h}_{R, f}^H \mathbf{\Theta_t}\textbf{h}_{1, R}$ and $H_{2, f}=h_{2, f}+\textbf{h}_{R, f}^H \mathbf{\Theta_t}\textbf{h}_{2, R}$, respectively. Given that non-coherent JT-CoMP is taken into consideration, the SINR and the corresponding achievable rate at U$_f$ can be expressed as~\cite{tanbourgi2014tractable, elhattab2022ris}:
\begin{gather}
    \label{eq:gamma_f}
    \gamma_{f}=\frac{\zeta_{1,f}P_1\abs{H_{1, f}}^2 + \zeta_{2,f}P_2\abs{H_{2, f}}^2}{\zeta_{c_1}P_1\abs{H_{1, f}}^2 + \zeta_{2,c}P_2\abs{H_{2, f}}^2 + \sigma^2}, \\
    \mathcal{R}_{f}=\log_2\left(1 + \gamma_f\right).
\end{gather}

\subsection{Outage Probability Analysis}
To further investigate the efficacy of strategically placing the STAR-RIS in improving the system performance, we analyze the outage probability experienced by cellular users. Following the principles of NOMA, $\forall i \in \mathcal{I}$, $c \in \mathcal{C}$, and $f \in \mathcal{F}$, if U$_{i,c}$ cannot decode $x_f$, or is capable of decoding $x_f$ but not $x_{i,c}$, an outage occurs, the probability of which is expressed as~\cite{9856598}:
\begin{equation}
    \mathcal{P}_{i,c} = \text{Pr}\,(\gamma_{i,c \rightarrow f} < \gamma_{th_f}) + \text{Pr}\,(\gamma_{i,c \rightarrow f} > \gamma_{th_f}, \gamma_{c} < \gamma_{th_{c}}),
\end{equation}
where $\gamma_{th_f}$ and $\gamma_{th_{c}}$ represent the outage thresholds for U$_f$ and U$_{i,c}$, respectively. Similarly, in regard to U$_f$, an outage occurs when it fails to decode x$_f$, and the corresponding outage probability is expressed as:
\begin{equation}
    \mathcal{P}_{f}=\text{Pr}\,(\gamma_{f}<\gamma_{th_f}).
\end{equation}

\vspace*{1.25em}
\LARGE{\textbf{Analytical Analysis}}
\normalsize

\section{End-to-end SINR Statistics}
\subsection{Effective Channel Characterization}
\label{sec:effective}
Let $Z_{i,u} = |{H}_{i,u}|^2 = (|h_{i,u}| + \sqrt{\beta_n} \sum_{k=1}^{K} |{h_{R,u}}||{h_{i,R}}|)^2$, where $n \in \{t, r\}$ represents the transmission and reflection regions of STAR-RIS, respectively. The distribution of $Z_{i,u}$ is derived in the following lemma.

\begin{lemma}
    \label{lem:effective}
    Assuming a large $K$, 
    and by applying MoM, the distribution of $Z_{i,u}$ is approximated as a Gamma distribution, $Z_{i,u} \sim \Gamma\big(k_{Z_{i,u}}, \theta_{Z_{i,u}}\big)$, with the following probability density function (PDF).
    \begin{equation}
        f_{Z_{i,u}}(x) = \frac{x^{k_{Z_{i,u}}-1} e^{-\frac{x}{\theta_{Z_{i,u}}}}}{\theta_{Z_{i,u}}^{k_{Z_{i,u}}} \Gamma\left(k_{Z_{i,u}}\right)},~x > 0,
    \end{equation}
    where $k_{Z_{i,u}} = \frac{\mu_{Z_{i,u}}^2}{\mu_{Z_{i,u}}^{(2)} - \mu_{Z_{i,u}}^2}$ and $\theta_{Z_{i,u}} = \frac{\mu_{Z_{i,u}}^{(2)} - \mu_{Z_{i,u}}^2}{\mu_{Z_{i,u}}}$ are the shape and scale parameters of the Gamma distribution, with $\mu_{Z_{i,u}}=\frac{2 K \sqrt{\beta \Omega_{iu} \Omega_{iR} \Omega_{Ru}}\,\Gamma_m \left(\frac{1}{2},\frac{1}{2},\frac{1}{2}\right)}{\sqrt{m_{iR} m_{Ru} m_{iu}}}+\beta  K^2 \Omega_{iR} \Omega_{Ru}+\Omega_{iu}$ and $\mu_{Z_{i,u}}^{(2)} = \frac{4 \beta^{3/2} K^3 \sqrt{\Omega_{iu}} (\Omega_{iR} \Omega_{Ru})^{3/2} \,\Gamma_m \left(\frac{1}{2},\frac{3}{2},\frac{3}{2}\right)}{\sqrt{m_{iu}} (m_{iR} m_{Ru})^{3/2}}+6\beta K^2 \Omega_{iR} \Omega_{iu} \Omega_{Ru}+\frac{4 K \Omega_{iu}^{3/2} \sqrt{\beta \Omega_{iR} \Omega_{Ru}}\,\Gamma_m \left(\frac{3}{2},\frac{1}{2},\frac{1}{2}\right)}{\sqrt{m_{iR} m_{Ru}}m_{iu}^{3/2}}+\frac{\beta^2 K^4 \Omega_{iR}^2
        (m_{iR}+1) (m_{Ru}+1) \Omega_{Ru}^2}{m_{iR}m_{Ru}}+\frac{(m_{iu}+1) \Omega_{iu}^2}{m_{iu}}$ as the first and second moments of $Z_{i,u}$, respectively, and $\Gamma_m\left(a,b,c\right)=\frac{\Gamma \left(m_{iu}\,+\,a\right) \Gamma \left(m_{iR}\,+\,b\right) \Gamma \left(m_{Ru}\,+\,c\right)
        }{\Gamma \left(m_{iu}\right) \Gamma \left(m_{iR}\right) \Gamma \left(m_{Ru}\right)}$.
\end{lemma}

\begin{proof}
    For brevity, let the combined channel be $G_{i,u} = \sqrt{\beta_n} \sum_{k=1}^{K} |{h_{R,u}}||{h_{i,R}}|$. By applying CLT, and noting that it is a scaled double-Nakagami random variable (RV), the distribution of $G_{i,u}$ can be approximated as a Gamma distribution, i.e.,
    \[G_{i,u} \sim \Gamma\left(\frac{\mu_{G_{i,u}}^2}{\mu_{G_{i,u}}^{(2)} - \mu_{G_{i,u}}^2}, \frac{\mu_{G_{i,u}}^{(2)} - \mu_{G_{i,u}}^2}{\mu_{G_{i,u}}}\right),\]
    where $\mu_{G_{i,u}}$ and $\mu_{G_{i,u}}^{(2)}$ are the first and second moments of $G_{i,u}$, respectively, with the $p$-th moment of $G_{i,u}$ given by~\cite{9290053}
    \begin{equation}
        \mu_{G_{i,u}}^{(p)} = \frac{(K\sqrt{\beta_n})^p({\Omega_{iR}}{\Omega_{Ru}})^{p/2}\,\Gamma\left(m_{Ru}+\frac{p}{2}\right)\Gamma\left(m_{iR}+\frac{p}{2}\right)}{({m_{Ru}}{m_{iR}})^{p/2}\,\Gamma\left(m_{iR}\right)\Gamma\left(m_{Ru}\right)},
    \end{equation}
    and as $|h_{i,u}| \sim Nakagami(m_{i,u}, \Omega_{i,u})$, the $p$-th moments are known to be given by $\mu_{|h_{i,u}|}^{(p)} = \frac{\Gamma\left(m_{iu}+\frac{p}{2}\right){\Omega_{iu}}^{p/2}}{\Gamma\left(m_{iu}\right){m_{iu}}^{p/2}}$.
    Since $|h_{i,u}|$ and $G_{i,u}$ are independent, the $p$-th moment of $|\textbf{H}_{i,u}|$ can be obtained via the moments of its summands, i.e., $|h_{i,u}|$ and $G_{i,u}$, by applying the binomial theorem. Hence, the $p$-th moment of $|\textbf{H}_{i,u}|$ is given by
    \begin{equation}
        \label{eq:binom}
        \mu_{|\textbf{H}_{i,u}|}^{(p)} = \sum_{q=0}^{p} \binom{\,p\,}{\,q\,}\mu_{|h_{i,u}|}^{(q)} \mu_{G_{i,u}}^{(p-q)}.
    \end{equation}
    Knowing that only the first two moments of $Z_{i,u}$, $\mu_{Z_{i,u}} = \mu_{|\textbf{H}_{i,u}|}^{(2)}$ and $\mu_{Z_{i,u}}^{(2)} = \mu_{|\textbf{H}_{i,u}|}^{(4)}$, are necessary to approximate its distribution as a Gamma distribution, the first two moments of $Z_{i,u}$ are, therefore, given by
    \begin{align}
        \label{eq:lem1_1}
        \mu_{Z_{i,u}}                         & = \mu_{|h_{i,u}|}^{(2)} + 2\mu_{|h_{i,u}|}\mu_{G_{i,u}} + \mu_{G_{i,u}}^{(2)},                                                                                        \\
        \label{eq:lem1_2} \mu_{Z_{i,u}}^{(2)} & =\mu_{|h_{i,u}|}^{(4)} + 4\mu_{|h_{i,u}|}^{(3)}\mu_{G_{i,u}} + 6\mu_{|h_{i,u}|}^{(2)}\mu_{G_{i,u}}^{(2)} + 4\mu_{|h_{i,u}|}\mu_{G_{i,u}}^{(3)} + \mu_{G_{i,u}}^{(4)}.
    \end{align}
    The final expression of moments can thus be obtained through the means of substitution in~(\ref{eq:lem1_1}) and~(\ref{eq:lem1_2}).
\end{proof}

Using the moment matching-based Gamma approximation approach, the distribution of the sum of a Gamma RV and the square of a Nakagami-$m$ RV with different shape and scale parameters, further weighted by different constant terms, i.e., path loss and power allocation factors, is derived in the following lemma.

\begin{lemma}
    \label{lem:sum}
    Let $\mathcal{B}_{i,u}^{(a,b)} = a Z_{i,u} + b |h_{i',u}|^2$, where $a, b \in \mathbb{R}^+$, $i' \in \mathcal{I} \setminus \{i\}$, then the distribution of $\mathcal{B}_{i,u}^{(a,b)}$ is approximated as a Gamma distribution, $\mathcal{B}_{i,u}^{(a,b)} \sim \Gamma\big(k_{\mathcal{B}_{i,u}^{(a,b)}}, \theta_{\mathcal{B}_{i,u}^{(a,b)}}\big)$, with the following PDF
    \begin{equation}
        f_{\mathcal{B}_{i,u}^{(a,b)}}(x) = \frac{x^{k_{\mathcal{B}_{i,u}^{(a,b)}}-1} e^{-\frac{x}{\theta_{\mathcal{B}_{i,u}^{(a,b)}}}}}{\theta_{\mathcal{B}_{i,u}^{(a,b)}}^{k_{\mathcal{B}_{i,u}^{(a,b)}}} \Gamma\Big(k_{\mathcal{B}_{i,u}^{(a,b)}}\Big)},~x > 0,
    \end{equation}
    where $k_{\mathcal{B}_{i,u}^{(a,b)}} = \frac{\mu_{\mathcal{B}_{i,u}^{(a,b)}}^2}{\mu_{\mathcal{B}_{i,u}^{(a,b)}}^{(2)} - \mu_{\mathcal{B}_{i,u}^{(a,b)}}^2}$ and $\theta_{\mathcal{B}_{i,u}^{(a,b)}} = \frac{\mu_{\mathcal{B}_{i,u}^{(a,b)}}^{(2)} - \mu_{\mathcal{B}_{i,u}^{(a,b)}}^2}{\mu_{\mathcal{B}_{i,u}^{(a,b)}}}$, with $\mu_{\mathcal{B}_{i,u}^{(a,b)}} = a \mu_{Z_{i,u}} + b \Omega_{i,u}$ and $\mu_{\mathcal{B}_{i,u}^{(a,b)}}^{(2)} = a^2 \mu_{Z_{i,u}}^{(2)} + 2 a b \mu_{Z_{i,u}} \Omega_{i,u} + b^2 \Omega_{i,u}^2 (1 + \frac{1}{m_{i,u}})$ as the first and second moments of $\mathcal{B}_{i,u}^{(a,b)}$, respectively.
\end{lemma}

\begin{proof}
    As $|h_{i,u}| \sim Nakagami(m_{i,u}, \Omega_{i,u})$, the square of $|h_{i,u}|$ is known to be Gamma distributed, i.e., $|h_{i,u}|^2 \sim \Gamma\left(k_{i,u}, \theta_{i,u}\right)$, where $k_{i,u} = m_{i,u}$ and $\theta_{i,u} = \frac{\Omega_{i,u}}{m_{i,u}}$ are the shape and scale parameters of the Gamma distribution, respectively. Further, by using the scaling property of Gamma distribution, i.e., $X \sim \Gamma\left(k, \theta\right) \implies aX \sim \Gamma\left(k, a\theta\right)$, the first and second moments of $a Z_{i,u}$ are given by $\mu_{a Z_{i,u}} = a \mu_{Z_{i,u}}$ and $\mu_{a Z_{i,u}}^{(2)} = a^2 \mu_{Z_{i,u}}^{(2)}$, respectively. Finally, the first and second moments of $\mathcal{B}_{i,u}^{(a,b)}$ can be obtained by applying the binomial theorem in~(\ref{eq:binom}) and substituting the moments of $a Z_{i,u}$ and $b |h_{i,u}|^2$.
\end{proof}

\subsection{Probability Density Functions of SINRs}
\begin{lemma}
    \label{lem:beta_prime_cf}
    The PDF of the SINR at U$_{c_{i}}$ to decode the signal of U$_{f}$, i.e., $\gamma_{i,c \rightarrow f}$, is given by
    \begin{equation}
        \label{eq:beta_prime_cf}
        f_{\gamma_{i,c \rightarrow f}}(x) = \frac{\theta_{\mathcal{W}_{i,c,f}} \left(\frac{x \theta_{\mathcal{W}_{i,c,f}}}{\rho\zeta_{i,f}\theta_{Z_{i,c}}}\right)^{\alpha_{i,c}} \left(\frac{x \theta_{\mathcal{W}_{i,c,f}}}{\rho\zeta_{i,f}\theta_{Z_{i,c}}}+1\right)^{\nu_{i,c\rightarrow f}}}{\rho\zeta_{i,f}\theta_{Z_{i,c}} B(k_{Z_{i,c}},k_{\mathcal{W}_{i,c,f}})},
    \end{equation}
    for $x>0$, where $\alpha_{i,c}=k_{Z_{i,c}}-1$, $\nu_{i,c\rightarrow f}=-(k_{Z_{i,c}}+k_{\mathcal{W}_{i,c,f}})$, $B\left(\cdot\:,\cdot\right)$ is the Euler Beta function, $k_{\mathcal{W}_{i,c,f}} = \frac{\mu_{\mathcal{W}_{i,c,f}}^2}{\mu_{\mathcal{W}_{i,c,f}}^{(2)} - \mu_{\mathcal{W}_{i,c,f}}^2}$, and $\theta_{\mathcal{W}_{i,c,f}} = \frac{\mu_{\mathcal{W}_{i,c,f}}^{(2)} - \mu_{\mathcal{W}_{i,c,f}}^2}{\mu_{\mathcal{W}_{i,c,f}}}$, with $\mu_{\mathcal{W}_{i,c,f}} = \mu_{\mathcal{B}_{i,c}^{(\rho\zeta_{i,c},\,\rho)}} + 1$ and $\mu_{\mathcal{W}_{i,c,f}}^{(2)} = \mu_{\mathcal{B}_{i,c}^{(\rho\zeta_{i,c},\,\rho)}}^{(2)} + 2\mu_{\mathcal{B}_{i,c}^{(\rho\zeta_{i,c},\,\rho)}} + 1$.
\end{lemma}

\begin{proof}
    The expression in~(\ref{eq:gamma_icf}) can be rewritten as $\gamma_{i,c \rightarrow f} = \frac{\rho\zeta_{i,f} Z_{i,c}}{\mathcal{W}_{i,c,f}}$, where $\mathcal{W}_{i,c,f} = \mathcal{B}_{i,c}^{(\rho\zeta_{i,c},\,\rho)} + 1$, with $Z_{i,c}$ and $\mathcal{B}_{i,c}^{(\rho\zeta_{i,c},\,\rho)}$ both being Gamma RVs based on the statistics derived earlier. Then, the distribution of $\mathcal{W}_{i,c,f}$ can also be approximated by an equivalent Gamma RV, i.e., $\mathcal{W}_{i,c,f} \sim \Gamma\big(k_{\mathcal{W}_{i,c,f}}, \theta_{\mathcal{W}_{i,c,f}}\big)$. As $Z_{i,c}$ and $\mathcal{W}_{i,c,f}$ are two independent Gamma RVs, the ratio of two Gamma RVs is known to follow a Beta prime distribution, i.e., $\frac{\rho\zeta_{i,f}Z_{i,c}}{\mathcal{W}_{i,c,f}} \sim \beta'\big(k_{Z_{i,c}}, k_{\mathcal{W}_{i,c,f}}, 1, \rho\zeta_{i,f}\theta_{Z_{i,c}}/{\theta_{\mathcal{W}_{i,c,f}}}\big)$, corresponding to the PDF in~(\ref{eq:beta_prime_c}).
\end{proof}

\begin{corollary}
    As $\gamma_{i,c \rightarrow f}$ and $\gamma_{i,c}$ are closely related, differing only by weighting constants, the PDF of $\gamma_{i,c}$ is given by
    \begin{equation}
        \label{eq:beta_prime_c}
        f_{\gamma_{i,c}}(x) = \frac{\theta_{\mathcal{W}_{i,c}} \left(\frac{x \theta_{\mathcal{W}_{i,c}}}{\rho\zeta_{i,c}\theta_{Z_{i,c}}}\right)^{\alpha_{i,c}} \left(\frac{x \theta_{\mathcal{W}_{i,c}}}{\rho\zeta_{i,c}\theta_{Z_{i,c}}}+1\right)^{\nu_{i,c}}}{\rho\zeta_{i,c}\theta_{Z_{i,c}} B(k_{Z_{i,c}},k_{\mathcal{W}_{i,c}})},~x>0,
    \end{equation}
    where $\nu_{i,c}=-(k_{Z_{i,c}}+k_{\mathcal{W}_{i,c}})$, $k_{\mathcal{W}_{i,c}} = \frac{\mu_{\mathcal{W}_{i,c}}^2}{\mu_{\mathcal{W}_{i,c}}^{(2)} - \mu_{\mathcal{W}_{i,c}}^2}$, and $\theta_{\mathcal{W}_{i,c}} = \frac{\mu_{\mathcal{W}_{i,c}}^{(2)} - \mu_{\mathcal{W}_{i,c}}^2}{\mu_{\mathcal{W}_{i,c}}}$, with $\mu_{\mathcal{W}_{i,c}} = \mu_{\mathcal{B}_{i,c}^{(0,\,\rho)}} + 1$ and $\mu_{\mathcal{W}_{i,c}}^{(2)} = \mu_{\mathcal{B}_{i,c}^{(0,\,\rho)}}^{(2)} + 2\mu_{\mathcal{B}_{i,c}^{(0,\,\rho)}} + 1$ as the first and second moment of $\mathcal{W}_{i,c} = \mathcal{B}_{i,c}^{(0,\,\rho)} + 1$, respectively.
\end{corollary}

\begin{lemma}
    \label{lem:beta_prime_f}
    The PDF of the SINR at $U_f$, i.e., $\gamma_{f}$, is given by
    \begin{equation}
        \label{eq:beta_prime_f}
        f_{\gamma_{f}}(x) = \frac{\theta_{\mathcal{W}_{f}} \left(\frac{x \theta_{\mathcal{W}_{f}}}{\theta_{\mathcal{V}_{f}}}\right)^{\alpha_{f}} \left(\frac{x \theta_{\mathcal{W}_{f}}}{\theta_{\mathcal{V}_f}}+1\right)^{\nu_{f}}}{\theta_{\mathcal{V}_f} B(k_{\mathcal{V}_f},k_{\mathcal{W}_f})},~x>0
    \end{equation}
    where $B\left(\cdot\:,\cdot\right)$ is the Euler Beta function, $\nu_{f}=-(k_{\mathcal{V}_{f}}+k_{\mathcal{W}_{f}})$, $k_{\mathcal{V}_{f}} = \frac{\mu_{\mathcal{V}_{f}}^2}{\mu_{\mathcal{V}_{f}}^{(2)} - \mu_{\mathcal{V}_{f}}^2}$, $\theta_{\mathcal{V}_{f}} = \frac{\mu_{\mathcal{V}_{f}}^{(2)} - \mu_{\mathcal{V}_{f}}^2}{\mu_{\mathcal{V}_{f}}}$, $k_{\mathcal{W}_{f}} = \frac{\mu_{\mathcal{W}_{f}}^2}{\mu_{\mathcal{W}_{f}}^{(2)} - \mu_{\mathcal{W}_{f}}^2}$, and $\theta_{\mathcal{W}_{f}} = \frac{\mu_{\mathcal{W}_{f}}^{(2)} - \mu_{\mathcal{W}_{f}}^2}{\mu_{\mathcal{W}_{f}}}$, with $\mu_{\mathcal{V}_{f}} = \rho (\zeta_{i,f} \mu_{Z_{i,f}} + \zeta_{i',f} \mu_{Z_{i',f}})$, $\mu_{\mathcal{V}_{f}}^{(2)} = \rho^2 (\zeta_{i,f}^{2} \mu_{Z_{i,f}}^{(2)}+ 2 \zeta_{i,f} \zeta_{i',f} \mu_{Z_{i,f}} \mu_{Z_{i',f}} + \zeta_{i',f}^{2} \mu_{Z_{i',f}}^{(2)})$, $\mu_{\mathcal{W}_{f}} = \rho(\zeta_{i,c} \mu_{Z_{i,f}} + \zeta_{i',c} \mu_{Z_{i',f}}) + 1$, and $\mu_{\mathcal{W}_{f}}^{(2)} = 2 \mu_{Z_{i,f}} (\rho^2 \zeta_{i,c} \zeta_{i',c} \mu_{Z_{i',f}}+\rho \zeta_{i',c}+\rho \zeta_{i,c})+\rho^2 \zeta_{i',c}^2 \mu_{Z_{i',f}}^2+\rho^2 \zeta_{i,c}^2\mu_{Z_{i,f}}^2+1$.
\end{lemma}

\begin{proof}
    The expression in~(\ref{eq:gamma_f}) can be rewritten as $\gamma_{f} = \frac{\mathcal{V}_{f}}{\mathcal{W}_{f}}$, where $\mathcal{V}_{f} = \rho \zeta_{i,f} Z_{i,f} + \rho \zeta_{i',f} Z_{i',f}$ and $\mathcal{W}_{f} = \rho \zeta_{i,c} Z_{i,f} + \rho \zeta_{i',c} Z_{i',f} + 1$, with $Z_{i,f}$ and $Z_{i',f}$ both being Gamma RVs. The rest of the proof is similar to that of Figure~\ref{lem:beta_prime_cf}.
\end{proof}

\begin{figure}[h!]
    \centering
    {\subfigure[The PDF of $\gamma_{i,c}$
        ]{\includegraphics[width=0.35\columnwidth]{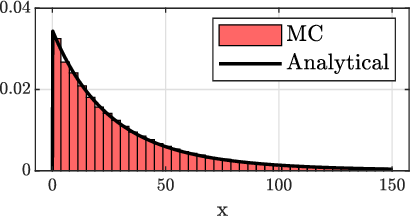}}
        \label{fig:SIM2A}}
    \subfigure[The CDF of $\gamma_{i,c}$
    ]{\includegraphics[width=0.35\columnwidth]{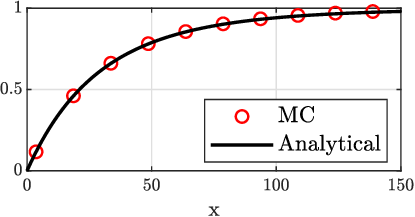}}
    \label{fig:SIM2E}
    \vspace{-5pt}
    \subfigure[
        The PDF of $\gamma_{i,c \rightarrow f}$
    ]{\includegraphics[width=0.35\columnwidth]{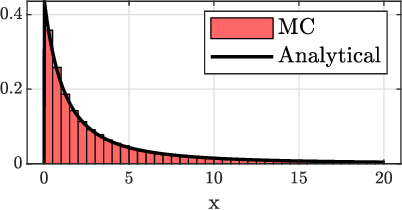}}
    \label{fig:SIM2B}
    \subfigure[
        The CDF of $\gamma_{i,c\rightarrow f}$
    ]{\includegraphics[width=0.35\columnwidth]{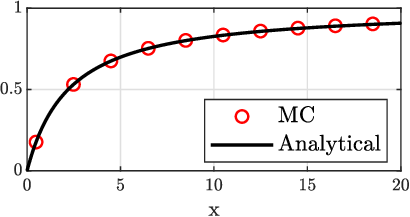}}
    \label{fig:SIM2F}
    \subfigure[
        The PDF of $\gamma_{f}$
    ]{\includegraphics[width=0.35\columnwidth]{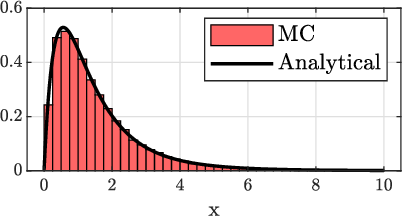}}
    \label{fig:SIM2C}
    \subfigure[
        The CDF of $\gamma_{f}$
    ]{\includegraphics[width=0.35\columnwidth]{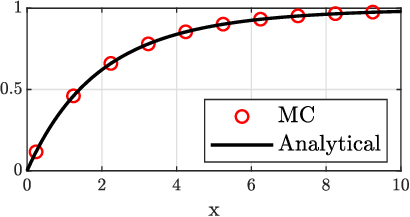}}
    \label{fig:SIM2G}
    \vspace{0.5em}
    \caption{The PDFs and CDFs of the SINRs at the center and edge user, with $K=34$ elements, $m_{i,u}=m_{i',u}=1$, and $m_{i,R}=m_{R,u}=2$, $\forall i \in \mathcal{I}, i' \in \mathcal{I} \setminus \{i\}$, $\forall u \in \mathcal{U}$.} \label{fig:pdfs} \vspace{-10pt}
\end{figure}

In Figure~\ref{fig:pdfs}, we compare the analytical and simulated PDFs and CDFs of SINRs of network users. The close alignment between the analytical approximations and Monte Carlo (MC) simulations affirms the accuracy of the derived expressions.
Additionally, we conduct the \textit{Kolmogorov-Smirnov (KS) goodness-of-fit test}, however, the details are omitted due to space constraints.

\section{Ergodic Rate (ER)}
The ER of the edge user is defined as
\begin{equation}
    \label{eq:ergodic_f}
    \mathcal{R}_{f} = \int_{0}^{\infty} \log_2(1 + x) f_{\gamma_{f}}(x) dx,
\end{equation}
where $f_{\gamma_{f}}(x)$ is the PDF of $\gamma_{f}$ in~(\ref{eq:beta_prime_f}). Noting that the PDF of $\gamma_{f}$ is a Beta prime distribution, the ER of the edge user is derived in the following theorem.
\begin{theorem}
    \label{thm:ergodic_f}
    In the proposed system model, the ER of the edge user is given by
    \begin{equation}
        \label{eq:ergodic_f_final}
        \mathcal{R}_{f} = \frac{1}{\ln(2)\Lambda_f}G_{3,3}^{3,2}\left(\frac{\theta_{\mathcal{W}_f}}{\theta_{\mathcal{V}_f}}\,\bigg\vert
        \begin{array}{c}
            0,1-k_{\mathcal{W}_f},1 \vspace*{-4pt} \\
            0,0,k_{\mathcal{V}_f} \vspace*{2.5pt}
        \end{array}
        \right),
    \end{equation}
    where $G_{p,q}^{m,n}\Big(z\,\Big\vert\begin{array}{c} a_1,\ldots,a_p \vspace*{-6.5pt} \\ b_1,\ldots,b_q \vspace*{2.5pt} \end{array}\Big)$ is the Meijer G-function, $\Lambda_f=B(k_{\mathcal{V}_f},k_{\mathcal{W}_f})\,\Gamma(\kappa_f)$ and $\kappa_f=k_{\mathcal{V}_f} +\,k_{\mathcal{W}_f}$.
\end{theorem}
\begin{proof}
    Substituting the PDF of $\gamma_{f}$ in~(\ref{eq:beta_prime_f}) into~(\ref{eq:ergodic_f}), we obtain
    \begin{equation}
        \label{eq:ergodic_f_1}
        \mathcal{R}_{f} = \frac{\theta_{\mathcal{W}_{f}}}{\theta_{\mathcal{V}_f} B(k_{\mathcal{V}_f},k_{\mathcal{W}_f})} \int_{0}^{\infty} \log_2(1 + x) \left(\frac{x \theta_{\mathcal{W}_{f}}}{\theta_{\mathcal{V}_f}}\right)^{\alpha_{f}} \times \left(1+\frac{x \theta_{\mathcal{W}_{f}}}{\theta_{\mathcal{V}_f}}\right)^{\nu_{f}} dx.
    \end{equation}
    From~\cite[Eq. (11)]{adamchik1990algorithm} and~\cite[Eq. (10)]{adamchik1990algorithm}, the logarithmic and power functions can be expressed in terms of a Meijer G-function, i.e., $\log_2(1 + z) = G_{2,2}^{1,2}\Big(z\,\Big\vert\begin{array}{c} 1,1 \vspace*{-6.5pt} \\ \vspace*{2.5pt} 1,0 \end{array}\Big) \Big/\ln(2)$ and $(1+z)^{\nu} = G_{1,1}^{1,1}\Big(z\,\Big\vert\begin{array}{c} 1-\nu \vspace*{-6.5pt} \\ \vspace*{2.5pt} 0 \end{array}\Big) \Big/\Gamma(\nu)$, respectively. Furthermore, using the analytical continuation of the Meijer G-function, the integral in~(\ref{eq:ergodic_f_1}) can be rewritten as
    \begin{align}
        \label{eq:ergodic_f_2}
        \mathcal{R}_{f} & = \frac{\theta_{\mathcal{W}_{f}}}{\ln(2)\theta_{\mathcal{V}_f} B(k_{\mathcal{V}_f},k_{\mathcal{W}_f})\Gamma(\kappa_f)} \int_{0}^{\infty} G_{2,2}^{1,2}\left(x\,\bigg\vert\begin{array}{c} 1,1 \vspace*{-4pt} \\ \vspace*{2.5pt} 1,0 \end{array}\right) \nonumber \\
                        & \quad \times G_{1,1}^{1,1}\left(\frac{x \theta_{\mathcal{W}_{f}}}{\theta_{\mathcal{V}_f}}\,\bigg\vert\begin{array}{c} -k_{\mathcal{W}_{f}} \vspace*{-4pt} \\ \vspace*{2.5pt} k_{\mathcal{V}_{f}} - 1 \end{array}\right) dx.
    \end{align}
    Finally, using the integral representation of the Meijer G-function, we obtain~(\ref{eq:ergodic_f_final}).
\end{proof}

To gain further insight, we express the high-SNR approximation for the ER of the edge user as
\begin{equation}
    \label{eq:ergodic_f_approx}
    \mathcal{R}_{f}^{\infty} \approx \frac{1}{\ln(2)\Lambda_{\tilde{f}}}G_{3,3}^{3,2}\left(\theta_{\tilde{f}}\,\bigg\vert
    \begin{array}{c}
        0,1-k_{\mathcal{\tilde{V}}_f},1 \vspace*{-4pt} \\
        0,0,k_{\mathcal{V}_f} \vspace*{2.5pt}
    \end{array}
    \right),
\end{equation}
where $\Lambda_{\tilde{f}}=B(k_{\mathcal{V}_f},k_{\mathcal{\tilde{V}}_f})\,\Gamma(\kappa_{\tilde{f}})$, $\kappa_{\tilde{f}}=k_{\mathcal{V}_f} +\,k_{\mathcal{\tilde{V}}_f}$, and $\theta_{\tilde{f}} = \frac{\theta_{\mathcal{\tilde{V}}_f}}{\theta_{\mathcal{V}_f}}$. The approximate parameters, denoted as $k_{\mathcal{\tilde{V}}_{f}}$ and $\theta_{\mathcal{\tilde{V}}_f}$, can be computed using the first and second moment, that is, $\mu_{\mathcal{\tilde{V}}_{f}} = \rho (\zeta_{i,c} \mu_{Z_{i,f}} + \zeta_{i',c} \mu_{Z_{i',f}})$ and $\mu_{\mathcal{\tilde{V}}_{f}}^{(2)} = \rho^2 (\zeta_{i,c}^{2} \mu_{Z_{i,f}}^{(2)}+ 2 \zeta_{i,c} \zeta_{i',c} \mu_{Z_{i,f}} \mu_{Z_{i',f}} + \zeta_{i',c}^{2} \mu_{Z_{i',f}}^{(2)})$, respectively.

Likewise, the ER of the center users is defined as
\begin{equation}
    \label{eq:ergodic_c}
    \mathcal{R}_{i,c} = \int_{0}^{\infty} \log_2(1 + x) f_{\gamma_{i,c}}(x) dx,
\end{equation}
where $f_{\gamma_{i,c}}(x)$ is the PDF of $\gamma_{i,c}$ in~(\ref{eq:beta_prime_c}). The ER of the center users can then be derived as follows.
\begin{theorem}
    \label{thm:ergodic_c}
    The ER for the center users is given by
    \begin{equation}
        \label{eq:ergodic_c_final}
        \mathcal{R}_{i,c} = \frac{1}{\ln(2)\Lambda_{i,c}}G_{3,3}^{3,2}\left(\frac{\theta_{\mathcal{W}_{i,c}}}{\rho\zeta_{i,c}\theta_{Z_{i,c}}}\,\bigg\vert
        \begin{array}{c}
            0,1-k_{\mathcal{W}_{i,c}},1 \vspace*{-4pt} \\
            0,0,k_{Z_{i,c}} \vspace*{2.5pt}
        \end{array}
        \right),
    \end{equation}
    where $\Lambda_{i,c}=B(k_{Z_{i,c}},k_{\mathcal{W}_{i,c}})\,\Gamma(\kappa_{i,c})$ and $\kappa_{i,c}=k_{Z_{i,c}} +\,k_{\mathcal{W}_{i,c}}$.
\end{theorem}

\begin{proof}
    The proof closely follows that of Figure~\ref{thm:ergodic_f}.
\end{proof}

Similar insights can be derived for the ER of the center users as of the edge user, however, the details are omitted due to space constraints.

\section{Outage Probability (OP)}
The OP for the edge user is defined as the probability that the instantaneous SINR at the edge user to decode its own message is below a certain threshold, and can be expressed as $\mathcal{P}_{f} = \text{Pr}\,(\gamma_{f}<\gamma_{th_f})$, where $\gamma_{th_f} = 2^{\mathcal{R}_{th_f}} - 1$ is the target SINR with $\mathcal{R}_{th_f}$ being the target rate for edge users. As $\gamma_{f} \sim \beta'\big(k_{\mathcal{V}_f}, k_{\mathcal{W}_f}, 1, {\theta_{\mathcal{V}_f}}/{\theta_{\mathcal{W}_f}}\big)$, and the CDF of a Beta prime distribution is known to be an incomplete Beta function, the OP for the edge user can be expressed as
\begin{equation}
    \label{eq:outage_f}
    \mathcal{P}_{f} = \frac{\Gamma(k_{\mathcal{V}_f}+k_{\mathcal{W}_f})}{\Gamma(k_{\mathcal{V}_f}) \Gamma(k_{{\mathcal{W}_f}})} B_{\psi_f}(k_{\mathcal{V}_f},k_{\mathcal{W}_f}),
\end{equation}
where $\psi_f=\frac{\lambda_{th_f}\theta_{\mathcal{W}_{f}}}{\theta_{\mathcal{V}_{f}} + \lambda_{th_f}\theta_{\mathcal{W}_{f}}}$, and $B_{z}(\cdot\:,\cdot)$ is the incomplete Beta function. As the threshold ($\lambda_{th_f}$) tends towards infinity, the incomplete Beta function in~(\ref{eq:outage_f}) converges to the Euler Beta function, i.e., $B_{\psi_f}(k_{\mathcal{V}f},k_{\mathcal{W}f}) \rightarrow B(k_{\mathcal{V}f},k_{\mathcal{W}_f})$.

Similarly, with regards to center users, the OP is defined as the probability that the instantaneous SINR for decoding the user's own message or the message of the edge user falls below a certain threshold. Mathematically, it can be expressed as $\mathcal{P}_{i,c} \approx \text{Pr}\,(\gamma_{i,c \rightarrow f} < \gamma_{th_f}) + \text{Pr}\,(\gamma_{i,c \rightarrow f} > \gamma_{th_f}, \gamma_{c} < \gamma_{th_{c}})$, where $\gamma_{th_{c}} = 2^{\mathcal{R}_{th_{c}}} - 1$ is the target SINR with $\mathcal{R}_{th_{c}}$ being the target rate for center users, and the approximate symbol is due to the fact that the detection sequence is not of fully independent events. The first term in the sum expression, denoted here onwards as $\mathcal{P}_{i,c}^{(1)}$, takes on the same form as that of OP for the edge user in~(\ref{eq:outage_f}), except for parameters, i.e., $k_{\mathcal{V}_{f}} \rightarrow k_{Z_{i,c}}$, $k_{\mathcal{W}_{f}} \rightarrow k_{\mathcal{W}_{i,c,f}}$, $\theta_{\mathcal{V}_{f}} \rightarrow \rho \zeta_{i,f} \theta_{Z_{i,c}}$, and $\theta_{\mathcal{W}_{f}} \rightarrow \theta_{\mathcal{W}_{i,c,f}}$. Furthermore, let $\mathcal{P}_{i,c}^{(2)}=\text{Pr}\,(\gamma_{i,c \rightarrow f} > \gamma_{th_f}, \gamma_{c} < \gamma_{th_{c}})$, then, the second term in the sum expression becomes
\begin{equation}
    \label{eq:outage_c}
    \mathcal{P}_{i,c}^{(2)} = I_{\psi_{i,c\rightarrow f}}(k_{\mathcal{W}_{i,c,f}}, k_{Z_{i,c}})\:I_{\psi_{i,c}}(k_{Z_{i,c}}, k_{\mathcal{W}_{i,c}}),
\end{equation}
where $\psi_{i,c\rightarrow f} = {\frac{\rho \zeta_{i,f} \theta_{Z_{i,c}}}{\rho \zeta_{i,f} \theta_{Z_{i,c}} + \lambda_{th_f}\theta_{\mathcal{W}_{i,c,f}}}}$, $\psi_{i,c}=\frac{\lambda_{th_c}\theta_{\mathcal{W}_{i,c}}}{\rho \zeta_{i,c} \theta_{Z_{i,c}} + \lambda_{th_c}\theta_{\mathcal{W}_{i,c}}}$, and $I_{z}(\cdot\:,\cdot)$ is the regularized incomplete Beta function. The OP for the center user is then given by $\mathcal{P}_{i,c} \approx \mathcal{P}_{i,c}^{(1)} + \mathcal{P}_{i,c}^{(2)}$. Further improvement in approximation can be made by making use of the fact that outage performance cannot be better than that of the interference-free noise-only case. Therefore, the final expression of the OP can be expressed as the maximum of the two cases, i.e., $\overline{\mathcal{P}_{i,c}} \approx \max\{\mathcal{P}_{i,c},\,\text{Pr}\,(\gamma_{c} < \gamma_{th_{c}})\}$.

\vspace*{1.25em}
\LARGE{\textbf{Numerical Results}}
\normalsize

\begin{table}[b!]
    \centering
    \caption{Simulation Parameters}
    \label{tab:params}
    \resizebox{0.65\columnwidth}{!}{%
        \begin{tabular}{|c|c|}
            \hline
            \textbf{Parameter}                                  & \textbf{Value}                       \\
            \hline
            \hline
            Path-loss exponent of BS$_i$-(U$_{c_i}$, RIS) links & $\alpha_{i\rightarrow c}=3$          \\
            Path-loss exponent of BS$_i$-U$_f$ link             & $\alpha_{i\rightarrow f}=3.5$        \\
            Path-loss exponent BS$_i$-RIS links                 & $\alpha_{i\rightarrow R}=3$          \\
            Path-loss exponent of RIS-U$_{c_i}$ links           & $\alpha_{R\rightarrow c}=2.7$        \\
            Path-loss exponent of RIS-U$_f$ link                & $\alpha_{R\rightarrow f}=2.3$        \\
            Path-loss exponent of Interfering links             & $\alpha_{i\rightarrow c^\prime} = 4$ \\
            Rician factor of RIS-U$_{c_i}$ links                & $\kappa_{R\rightarrow c}=3$ dB       \\
            Rician factor of RIS-U$_f$ link                     & $\kappa_{R\rightarrow f}=4$ dB       \\ \hline
        \end{tabular}
    }
\end{table}

\section{Simulation Setup}
We consider an outdoor environment where the transmission bandwidth of the network is set to $B = 1$ MHz, and the power of AWGN is set to $\sigma^2 = -174 + 10\log_{10}(B)$ (dBm) with a noise figure $N_F$ of $12$ dB. For simplicity, we assume that the transmit powers of both BS$_1$ and BS$_2$ are identical, expressed as $P_1 = P_2 = P_t$. Moreover, the PA factors for U$_{1,c}$, U$_{2,c}$, and U$_f$ are fixed to $\zeta_{1,c}=\zeta_{2,c}=0.3$ and $\zeta_f=0.7$, respectively.

\begin{figure}[t!]
    \centering
    \includegraphics[width=0.65\columnwidth]{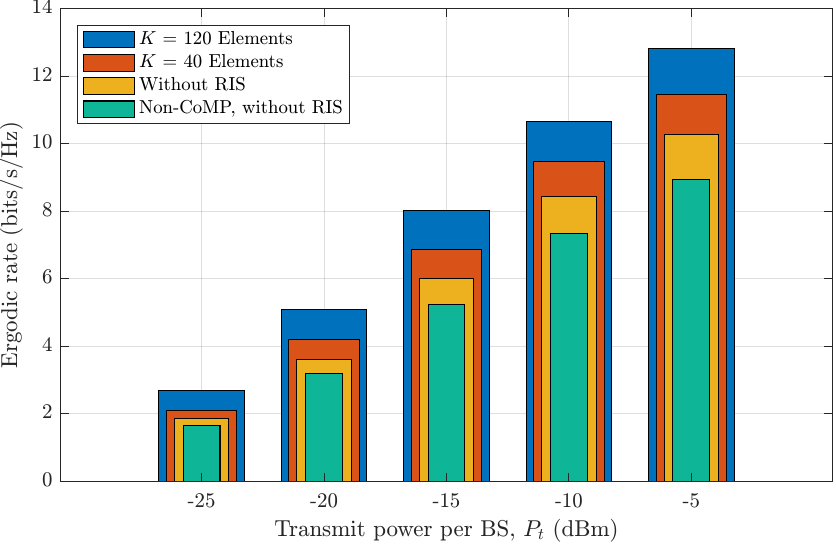}
    \vspace*{1em}
    \caption{The PDFs and CDFs of the SINRs at the center and edge user, with $K=34$ elements, $m_{i,u}=m_{i',u}=1$, and $m_{i,R}=m_{R,u}=2$, $\forall i \in \mathcal{I}, i' \in \mathcal{I} \setminus \{i\}$, $\forall u \in \mathcal{U}$.} \label{fig:ana_pdfs}
\end{figure}

In the three-dimensional Cartesian coordinate system, the locations of BS$_1$ and BS$_2$, each with a coverage radius of 60 m, are set to (-50m, 0m, 25m) and (50m, 0m, 25m) respectively. The STAR-RIS is strategically placed at the intersection of the two cells, near U$_f$, specifically, at the coordinates (0m, 25m, 5m). Additionally, the cellular users U$_{1,c}$, U$_{2,c}$, and U$_f$ are positioned at (-40m, 18m, 1m), (30m, 22m, 1m), and (0m, 35m, 1m), respectively. Some specific parameters used for simulation are outlined in Table~\ref{tab:params}.

\begin{figure}[b!]
    \centering
    \includegraphics[width=0.65\columnwidth]{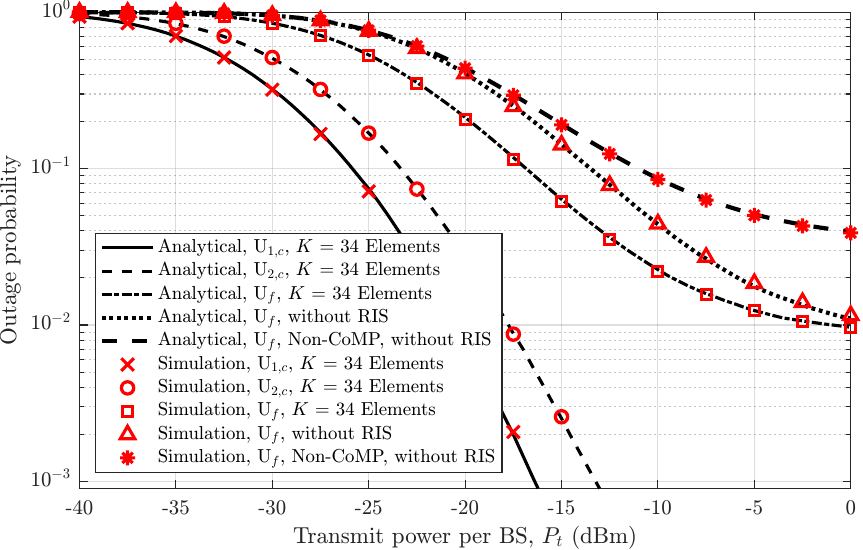}
    \vspace*{1em}
    \caption{Outage probability of network users versus $P_{t}$ for equal amplitude coefficients $(\beta^t=\beta^r)$, and element assignments $(\textbf{K}_R^1=\textbf{K}_R^2)$, when $K>0$.}
    \label{fig:ana_outage}
\end{figure}

\section{Improvements in Ergodic Rate}
The impact of the number of RIS elements on the ergodic rate of the network is shown in Figure~\ref{fig:ana_pdfs}. We observe that the ergodic rate increases with the number of RIS elements as the RIS elements amplify the channel links. Again, the STAR-RIS assisted CoMP-NOMA network outperforms other networks, due to  substantial diversity gains at U$_f$.

\section{Impact of RIS on Outage Probability}

In Figure~\ref{fig:ana_outage}, we assess the OP for all users under various transmit power levels ($P_t$) and system configurations, with fixed thresholds $\lambda_{th_f} = \lambda_{th_c} = 0$ dB. The STAR-RIS assisted CoMP-NOMA network demonstrates significant OP enhancements for U$_f$, attributed to the formation of vLOS paths. U$_{c_1}$ and U$_{c_2}$ show marginal improvements due to their predominant reliance on LOS paths from their corresponding BS. Notably, in the absence of CoMP, U$_f$ contends with elevated ICI and consistently high outage probabilities across all power levels. Moreover, we attribute the convergence of outage probability of U$_f$ towards the asymptotes to the problem of saturation inherent in NOMA.

\section{Exhaustive Search for Optimality Regions}
Finally, in Figure~\ref{fig:ana_exhaustive}, we demonstrate the effect of varying the RIS element assignments $(\textbf{K}_R^1,\,\textbf{K}_R^2)$ to BS$_1$ and BS$_2$, respectively, and the amplitude adjustments $(\beta_t,\,\beta_r)$ in an exhaustive fashion. Notably, the ergodic rate peaks when $\beta_t>\beta_r$ as a result of close proximity of the STAR-RIS to U$_f$, located within the transmission region of the RIS, thereby defining the optimal configuration for the network. Investigating optimization techniques for STAR-RIS resources can provide further insights to enhance spectral efficiency.

\begin{figure}[h!]
    \centering
    \includegraphics[width=0.65\columnwidth]{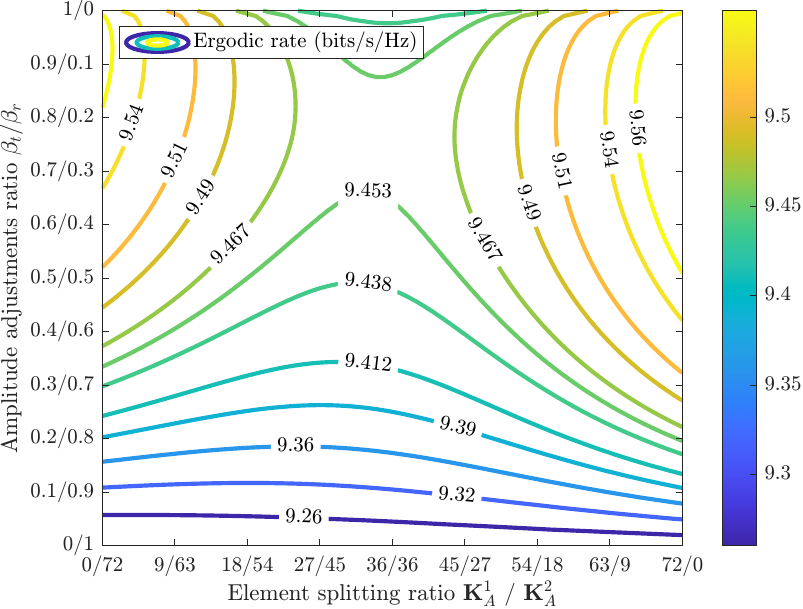}
    \vspace*{1em}
    \caption{Ergodic rate for varying RIS element assignments $(\textbf{K}_R^1,\,\textbf{K}_R^2)$ and amplitude adjustments $(\beta_t,\,\beta_r)$, with $P_t =-10$ dBm.}
    \label{fig:ana_exhaustive}
\end{figure}
\chapter{Energy Efficient Design for CoMP-NOMA Networks}
\label{chap:ee}

The convergence of CoMP and NOMA techniques has emerged as a promising solution to enhance both coverage and capacity within cellular networks. Recent research has further explored the integration of Reconfigurable Intelligent Surfaces (RIS) into CoMP-NOMA systems, demonstrating the potential for significant performance gains. By strategically deploying RIS, particularly at the cell edge, signal quality for edge users can be dramatically improved, leading to enhanced network coverage and capacity. However, despite these promising advancements, a crucial question remains: how do varying RIS configurations, the number of cooperating base stations (BSs), and the number of RIS elements influence the overall energy efficiency of the network? This chapter delves into this critical aspect, investigating the energy efficiency and passive beamforming (PBF) design within RIS-assisted CoMP-NOMA networks. We propose and analyze two distinct RIS configurations: Enhancement-only PBF (EO) and Enhancement \& Cancellation PBF (EC). The EO configuration focuses on optimizing RIS phases to solely enhance the desired signal quality for edge users, while the EC configuration aims to optimize RIS phases for both signal enhancement and interference suppression.

Through a comprehensive analysis, we explore the impact of these configurations, along with the number of cooperating BSs and RIS elements, on the network's energy efficiency. Additionally, we formulate a PBF design problem with the objective of maximizing energy efficiency through optimized RIS phase shifts. Our findings demonstrate that the integration of RIS into CoMP-NOMA networks offers substantial improvements in energy efficiency and overall network performance, highlighting its potential for the future of wireless communication systems.

\LARGE{\textbf{Towards Generalizability}}
\normalsize

\section{System Model}

We consider a downlink transmission scenario in an RIS-assisted multi-cell CoMP-NOMA network, as illustrated in Figure~\ref{fig:eff_system}. The network is comprised of $I$ cells, each modeled as a disk with radius $R_i$ and served by a BS located at its center, denoted by BS$_i$, where $i \in \mathcal{I} \triangleq \{1, 2, \ldots, I\}$. Each single-antenna BS utilizes two-user NOMA to serve user clusters, also equipped with single antennas, within its coverage area.\footnote{Due to processing complexity, latency of SIC at the receivers, and the practical limitations resulting in SIC error propagation, two-user NOMA pairs are considered in this work. Moreover, two-user configurations are of practical interest and have been standardized in 3GPP Release 15~\cite{ding2020simple, 10464446}.}

\begin{figure}[h!]
  \centering
  \includegraphics[width=0.65\columnwidth]{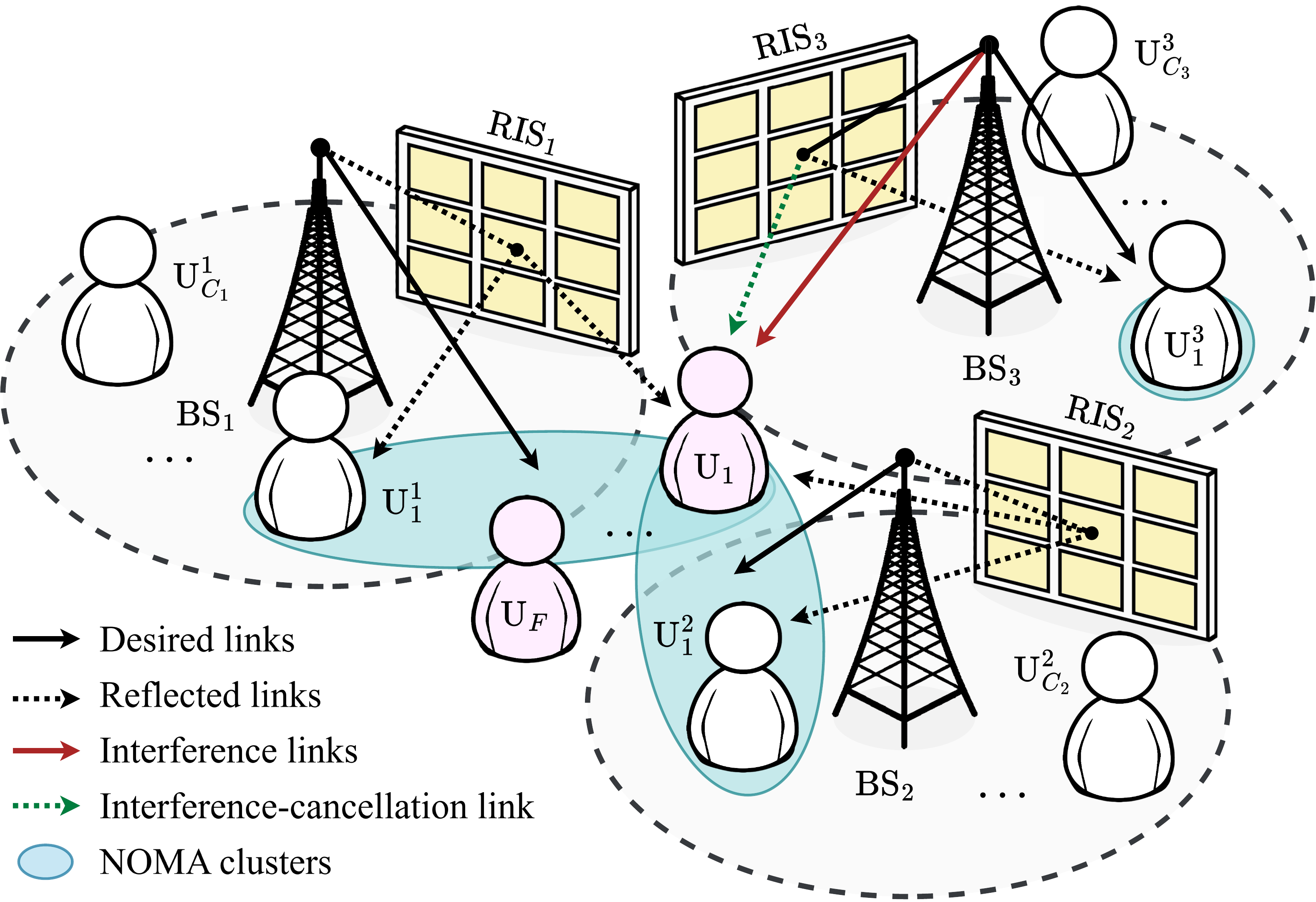}
  \vspace*{0.5em}
  \caption{An illustration of the RIS-assisted multi-cell CoMP-NOMA network.}
  \label{fig:eff_system}
\end{figure}

We define two user classes based on their location: the cell-center users and the edge users. The cell-center users reside within the disk of their associated cell, while the edge users fall outside it. Let $\mathcal{C}^{(i)} \triangleq \{1, 2, \ldots, C_i\}$ represent the set of indices for cell-center users associated with BS$_i$, and $\mathcal{F} \triangleq \{1, 2, \ldots, F\}$ represent the set of indices for the shared edge users cooperatively served by multiple BSs.
We denote cell-center users as U$^i_c$ (superscript $i$ indicates the serving BS) and edge users as U$^e_f$ (superscript $e$ represents the edge class), with $c \in \mathcal{C}^{(i)}$ and $f \in \mathcal{F}$. Furthermore, $\mathcal{U} = \bigcup_{i \in \mathcal{I}} \mathcal{C}^{(i)} \cup \mathcal{F}$ represents the index set for all users in the network. For ease of exposition and without loss of generality, we consider a single edge user and a single cell-center user per cell in this work, i.e., $C_i = F = 1$, $\forall i \in \mathcal{I}$.

For coordinated operation, the BSs are assumed to be interconnected via a high-speed backhaul network to a central processing unit (CPU).
BSs participating in CoMP are referred to as cooperative BSs and are denoted by the set $\mathcal{J} \triangleq \{1, 2, \ldots, J\}$, where $J \leq I$. To further improve the signal quality for edge users, each BS$_i$ is equipped with an RIS, denoted R$_i$, strategically placed at the cell edge.

\subsection{Channel Model}
We consider both large-scale and small-scale fading effects for each communication link in the system. Due to large propagation distances and the presence of numerous scatterers, the direct links, i.e., the channels between the BSs and the users, are modeled as Rayleigh fading channels. We denote the channel between BS$_i$ and user U$^n_u$ as $h^n_{i, u}$, where $n \in \{\mathcal{I}, e\}$ for cell-center or edge user, and $u \in \mathcal{U}$. Mathematically, this can be expressed as
\begin{equation}
  h^n_{i, u} = \sqrt{\frac{\rho_o}{PL(d^n_{i, u})}} v^n_{i, u},
\end{equation}
where $\rho_o$ is the reference path loss at $1$ m, $PL(d^n_{i, u})$ is the large-scale path loss modeled as $PL(d^n_{i, u}) = (d^n_{i, u})^{-\alpha_n}$, such that $\alpha_n$ is the path loss exponent, $d^n_{i, u}$ is the distance between BS$_i$ and U$^n_u$, and $v^n_{i, u} \in \mathbb{C}^{1\times 1}$ is the small-scale Rayleigh fading coefficient with zero mean and unit variance.

In contrast to the Rayleigh fading experienced on direct links, the channels between BSs and RIS are modeled as Rician fading channels due to the presence of a dominant line-of-sight (LoS) component. The channel between BS$_i$ and R$_i$ is denoted $\mathbf{h}_{i, \text{R}_i}$ and can be expressed as
\begin{equation}
  \mathbf{h}_{i, \text{R}_i} = \sqrt{\frac{\rho_o}{PL(d_{i, \text{R}_i})}} \left( \sqrt{\frac{\kappa}{1 + \kappa}} \mathbf{g}^{\text{LoS}}_{i, \text{R}_i} + \sqrt{\frac{1}{1 + \kappa}} \mathbf{g}^{\text{NLoS}}_{i, \text{R}_i} \right),
\end{equation}
where $\kappa$ is the Rician factor, $\mathbf{g}^{\text{LoS}}_{i, \text{R}_i} \in \mathbb{C}^{K\times 1}$ is the LoS channel vector given by
\begin{equation*}
  \mathbf{g}^{\text{LoS}}_{i, \text{R}_i} = \left[1, \ldots, e^{j(k-1)\pi\sin(\omega_i)}, \ldots, e^{j(K-1)\pi\sin(\omega_i)}\right]^T,
\end{equation*}
where $k \in \{1, 2, \ldots, K\}$ indexes elements of R$_i$ and $\omega_i$ is the angle of arrival (AoA) of the LoS component at R$_i$ while $\mathbf{g}^{\text{NLoS}}_{i, \text{R}_i} \in \mathbb{C}^{K\times 1}$ is the NLoS component which follows Rayleigh fading as previously described. The channel between RIS and edge users can be modeled similarly using Rician fading.

For the sake of simplicity, we assume perfect channel state information (CSI) at the BSs. While achieving perfect CSI in practice can be challenging, recent advancements in channel estimation techniques for RIS-assisted wireless networks have demonstrated the potential for accurate CSI acquisition with a reasonable overhead~\cite{zheng2022survey, wei2021channel, shtaiwi2021channel, zhou2022channel}.

\subsection{RIS Configuration}
Each BS$_i$ utilizes a passive RIS R$_i$ at the cell edge to improve signal quality or suppress interference for U$^e_f$, $\forall f \in \mathcal{F}$. RIS elements can independently adjust the phase shift of the incident signal and are assumed to be controlled by the CPU. Furthermore, the phase shift (PS) matrix associated with R$_i$ is expressed as $\mathbf{\Theta}_i = \text{diag}(l_1 e^{j\theta^i_1}, l_2 e^{j\theta^i_2}, \ldots, l_K e^{j\theta^i_K})$, where $l_k \in (0, 1]$ is the amplitude adjustment factor and $\theta^i_k \in [-\pi, \pi)$ is the phase shift of the $k$-th element. In this work, we assume an ideal RIS with perfect phase control and all reflection elements having a unit amplitude ($l_k = 1, \forall k$).\footnote{While practical RIS implementations are subject to phase quantization errors, this work prioritizes establishing a proof-of-concept for the benefits of RIS-assisted CoMP-NOMA in terms of energy efficiency and performance.
}

\section{Performance Analysis}

\subsection{Rate and Outage Probability Analysis}
According to the NOMA principle, the BSs serve multiple users simultaneously by superimposing their signals. Specifically, the signal transmitted by BS$_i$ can be expressed as $x_i = \sqrt{\zeta_i P_i} s_e + \sqrt{(1 - \zeta_i) P_i} s_{c_i}$, where $P_i$ is the transmit power of BS$_i$, $s_e$ and $s_{c_i}$ are the signals intended for U$^e_f$ and U$^i_c$, respectively, and $\zeta_i$ is the power allocation factor for the edge users. To ensure successful decoding by U$^i_c$, i.e., the strong user, $\zeta_i$ is constrained to $0.5 < \zeta_i < 1$~\cite{obeed2020user}.

The received signal at U$^e_f$ can be written as
\begin{equation}
  y^e_f = \underbrace{\sum_{j \in \mathcal{J}} H^e_{j, f} \sqrt{\zeta_j P_j} s_e}_{\text{CoMP gain}} + \underbrace{\sum_{j \in \mathcal{J}} H^e_{j, f} \sqrt{(1 - \zeta_j) P_j} s_{c_j}}_{\text{intra-cluster interference}} + \underbrace{\sum_{m \in \mathcal{I} \setminus \mathcal{J}} H^e_{m, f} x_i}_{\text{inter-cell interference}} + n_o,
\end{equation}
where $H^e_{j, f} = h^e_{j, f} + \mathbf{h}^T_{\text{R}_j, f} \mathbf{\Theta}_j \mathbf{h}_{j, \text{R}_j}$ and $H^e_{m, f} = h^e_{m, f} + \mathbf{h}^T_{\text{R}_i, f} \mathbf{\Theta}_i \mathbf{h}_{m, \text{R}_i}$ are the effective channels between BS$_j$ and U$^e_f$ and between BS$_i$ and U$^e_f$, respectively, and $n_o \sim \mathcal{CN}(0, \sigma^2)$ is the additive white Gaussian noise. To minimize synchronization overhead, we employ non-coherent JT-CoMP, where the edge user U$^e_f$ combines signals from cooperative BSs without CSI exchange~\cite{tanbourgi2014tractable}. Therefore, the signal-to-interference-plus-noise ratio (SINR) at U$^e_f$ can be expressed as
\begin{equation}
  \gamma^e_f = \frac{\sum_{j \in \mathcal{J}} \zeta_j P_j |H^e_{j, f}|^2}{\sum_{j \in \mathcal{J}} (1 - \zeta_j) P_j |H^e_{j, f}|^2 + Y_e + \sigma^2},
\end{equation}
where $Y_f = \sum_{m \in \mathcal{I} \setminus \mathcal{J}} P_m |H^e_{m, f}|^2$ represents the inter-cell interference term.

On the other hand, the received signal at U$^i_c$ is given by
\begin{multline}
  y^i_c = h^i_{i, c} \sqrt{(1 - \zeta_i) P_i} s_{c_i} + \sum_{j \in \mathcal{J}} h^i_{j, c} \sqrt{\zeta_j P_j} s_e \\
  + \sum_{j \in \mathcal{J}, j \neq i} h^i_{j, c} \sqrt{(1 - \zeta_j) P_j} s_{c_j} + \sum_{m \in \mathcal{I} \setminus \mathcal{J}} h^i_{m, c} x_m + n_o.
\end{multline}
Based on the SIC principle, the SINR at U$^i_c$ for decoding the signal intended for U$^e_f$ is given by
\begin{equation}
  \gamma^i_{c\rightarrow f} = \frac{\sum_{j \in \mathcal{J}} \zeta_j P_j |h^i_{j, c}|^2}{\sum_{j \in \mathcal{J}} (1 - \zeta_j) P_j |h^i_{j, c}|^2 + Y_i + \sigma^2},
\end{equation}
and the SINR at U$^i_c$ for decoding its own signal is
\begin{equation}
  \gamma^i_c = \frac{(1 - \zeta_i) P_i |h^i_{i, c}|^2}{\sum_{j \in \mathcal{J}, j \neq i} (1 - \zeta_j) P_j |h^i_{j, c}|^2 + Y_i + \sigma^2}.
\end{equation}
It is worth noting that due to their placement at the cell edge, the impact of RIS on the channels experienced by U$^i_c$ is negligible. Thus, the SINR expressions for U$^i_c$ only consider the direct links between the BSs and the users. Finally, the achievable rates for U$^e_f$ and U$^i_c$ can be calculated as
\begin{align}
  \mathcal{R}^e_f & = \log_2(1 + \gamma^e_f), 
\end{align}
and
\begin{align}
  \mathcal{R}^i_{c} & = \log_2(1 + \gamma^i_{c}).
\end{align}

An outage event occurs for cell-center users if U$^i_c$ fails to decode $s_e$ or is capable of decoding $s_e$ but fails to decode $s_{c_i}$. The corresponding outage probability can be defined as $\mathbb{P}^i_c = 1 - \mathbb{P}(\gamma^i_{c\rightarrow f} > \hat{\gamma_f}, \gamma^i_c > \hat{\gamma_c})$, where $\hat{\gamma_f}=2^{\mathcal{R}^e_{th}}-1$ and $\hat{\gamma_c}=2^{\mathcal{R}^i_{th}}-1$ represent the target SINR thresholds for U$^e_f$ and U$^i_c$, respectively, corresponding to their target rates $\mathcal{R}^e_{th}$ and $\mathcal{R}^i_{th}$. For edge user U$^e_f$, an outage occurs if it fails to decode $s_e$, and the outage probability can be formulated as
\begin{align}
  \mathbb{P}^e_f = \mathbb{P}(\gamma^e_f < \hat{\gamma_f}).
\end{align}

\subsection{Energy Efficiency}
We define the energy efficiency of the network as the ratio of the outage sum rate to the total power consumption. This metric essentially assesses whether the increase in outage sum rate resulting from CoMP outweighs the corresponding increase in total power consumption. Mathematically, the energy efficiency is formulated as
\begin{equation}
  \label{eq:eff}
  \eta_\text{EE} = \sum_{i \in \mathcal{I}} \frac{\sum_{c \in \mathcal{C}^{(i)}} \mathcal{R}^i_{\text{out}_c}}{\frac{1}{\lambda}P_i + P_Q} + \sum_{j \in \mathcal{J}} \frac{\sum_{f \in \mathcal{F}} \mathcal{R}^e_{\text{out}_f}}{\frac{1}{\lambda}P_j + P_Q + P_{\text{R}}},
\end{equation}
where $\mathcal{R}^i_{\text{out}_c} = (1 - \mathbb{P}^i_c) \mathcal{R}^i_{c}$ and $\mathcal{R}^e_{\text{out}_f} = (1 - \mathbb{P}^e_f) \mathcal{R}^e_f$ represent the effective outage rate for U$^i_c$ and U$^e_f$, respectively. Moreover, $P_Q$ represents the static power consumption of a cell, $P_{\text{R}}=KP_\text{ele}$ denotes the total power consumption of R$_i$, where $P_\text{ele}$ is the power consumption of $k$-th element, and $\lambda \in (0, 1]$ signifies the power amplifier efficiency.

  \section{Passive Beamforming Design}
  The primary objective of the PBF design is to optimize the energy efficiency of the network by strategically adjusting the RIS phase shifts. This involves configuring R$_j$ associated with cooperative BS$_j$ to enhance the signal quality for U$^e_f$, while simultaneously utilizing R$_k$, $\forall m \in \mathcal{I} \setminus \mathcal{J}$, to suppress the inter-cell interference experienced by U$^e_f$. The optimization problem is formulated as
  \begin{maxi!}|s|
  {\mathbf{\Phi}}{\text{energy efficiency}~\eta_\text{EE}~\text{in}~\eqref{eq:eff}}{\label{eq:opt}}{}
  \addConstraint{\theta^j_k}{\in [-\pi, \pi),}{\; \forall k \in [1, K],\, j \in \mathcal{J}} \label{eq:coop}
  \addConstraint{\theta^m_k}{\in [-\pi, \pi),}{\; \forall k \in [1, K],\, m \in \mathcal{I} \setminus \mathcal{J}} \label{eq:interf}
  \end{maxi!}
  where $\mathbf{\Phi} = \{\mathbf{\Theta}_1, \mathbf{\Theta}_2, \ldots, \mathbf{\Theta}_I\}$ represents the set of phase shift matrices for all RIS.

  For cooperative BSs, the phase shifts of R$_i$ in constraint~\eqref{eq:coop} are optimized to maximize the effective channel gain $|H^e_{j, f}|^2$. From~\cite{wu2019intelligent}, the optimal phase shift for each element of R$_j$ is given by
  \begin{equation}
    \theta^j_k = \arg(h^e_{j, f}) - \arg(h_{\text{R}_j, f}^{(k)} \cdot h_{j, \text{R}_j}^{(k)}),
  \end{equation}
  where $\arg(\cdot)$ denotes the argument function, and $h_{\text{R}_j, f}^{(k)}$ and $h_{j, \text{R}_j}^{(k)}$ represent the $k$-th elements of the channel vectors $\mathbf{h}_{\text{R}_j, f}$ and $\mathbf{h}_{j, \text{R}_j}$, respectively.

  In contrast to signal enhancement, the phase shifts of R$_m$ for non-cooperating BSs in constraint~\eqref{eq:interf} are adjusted to minimize the effective channel gain $|H^e_{m, f}|^2$, thus mitigating interference. Thus, the optimal phase shift for each element of R$_m$ can be calculated as
  \begin{equation}
    \theta^m_k = \text{mod}\left[\phi^m_k + \pi,\: 2\pi\right] - \pi,
  \end{equation}
  where $\phi^m_k = \arg(h^e_{m, f}) - \arg(h_{\text{R}_m, f}^{(k)} \cdot h_{m, \text{R}_m}^{(k)})$ and $\text{mod}[\cdot]$ denotes the modulo operation. Adding $\pi$ ensures a 180° phase shift, effectively suppressing interference.

  It should be noted that incorporating cell-center users into the optimization problem, or increasing the number of users per cell, significantly increases the complexity of the beamforming design. To maintain tractability, this work focuses on a single edge user and a single cell-center user per cell.

  \vspace*{1.25em}
  \LARGE{\textbf{Numerical Results}}
  \normalsize

  \section{Simulation Setup}
  The performance of the proposed design is evaluated in a network consisting of $I = 6$ cells, each with a radius of $R_i = 75$ m. All BSs transmit at the same power level, i.e., $P_i = P_t$ dBm, $\forall i \in \mathcal{I}$, and the power allocation factor for edge users is set to $\zeta_i = 0.7$. The RIS are positioned at the cell edge, resulting in a distance of $d_{i, \text{R}_i} = 75$ m between BS$_i$ and R$_i$. The distances between BSs and users are configured as follows: $d^i_{i, c} = 50$ m, $d^e_{i, f} = 150$ m, and $d^i_{k, c} = 200$ m for $k \neq i$, where $i \in \mathcal{I}$, $c \in \mathcal{C}^{(i)}$, and $f \in \mathcal{F}$. Similarly, the distance between RIS and edge users is set to $d^e_{i, f} = 75$ m.

  The path loss exponents are set to $\alpha_{\text{R}} = 2.7$, $\alpha_i = 3$, $\alpha_e = 3.5$, and $\alpha_{\text{ici}} = 4$ for RIS, BS, edge user, and inter-cell interference links, respectively. The network operates at a carrier frequency of $f_c = 2.4$ GHz, and the noise power is defined as $\sigma^2 = -174 + 10 \log_{10}(B)$, with a bandwidth $B = 10$ MHz. Table~\ref{tab:eff_params} summarizes the remaining simulation parameters.

  \begin{table}[h!]
    \centering
    \caption{Simulation Parameters}
    \label{tab:eff_params}
    \resizebox{0.65\columnwidth}{!}{%
      \begin{tabular}{|c|c|}
        \hline
        \textbf{Parameter}                                    & \textbf{Value} \\
        \hline
        \hline
        Power amplifier efficiency $\lambda$                  & $0.4$          \\
        Reference path loss $\rho_o$                          & $-30$ dBm      \\
        Static power consumption $P_Q$                        & $30$ dBm       \\
        Power dissipated at $k$-th RIS element $P_\text{ele}$ & $5$ dBm        \\
        Target cell-center user rate $\mathcal{R}^i_{th}$     & $1$ bps/Hz     \\
        Target edge user rate $\mathcal{R}^e_{th}$            & $0.5$ bps/Hz   \\
        Rician factor $\kappa$                                & $3$ dB         \\
        \# of channel realizations $N_{\text{mc}}$            & $10^4$         \\
        \hline
      \end{tabular}
    }
  \end{table}

  \begin{figure}[t!]
    \centering
    \includegraphics[width=0.65\columnwidth]{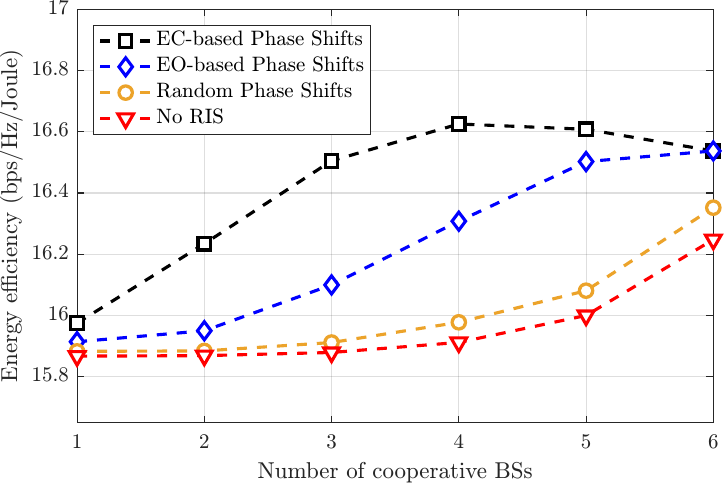}
    \vspace*{0.5em}
    \caption{Energy efficiency vs. number of cooperative BSs $J$ at $P_t = 0$ dBm and under various RIS configurations with $K = 70$ elements.} \label{fig:eff_eff}
  \end{figure}

  \section{Impact of Cooperation on Energy Efficiency}
  Figure~\ref{fig:eff_eff} illustrates the energy efficiency of the network as a function of the number of cooperating BSs $J$ for various RIS configurations. The investigated scenarios include: no RIS, RIS with random phase shifts, EO-based RIS, and EC-based RIS. We observe that for the EO and EC configurations, the energy efficiency initially increases with a growing number of cooperating BSs, reaching a peak at $J = 4$. Beyond this point, the efficiency experiences a decline due to the saturation of both the achievable rate and the outage probability, leading to diminishing returns. It should be noted that the EC configuration consistently outperforms other scenarios $\forall J$, except when $J = I$. In this particular case where all BSs are cooperative, interference cancellation becomes redundant, leading to equivalent performance between the EO and EC configurations.

  \begin{figure}[t!]
    \centering
    \includegraphics[width=0.65\columnwidth]{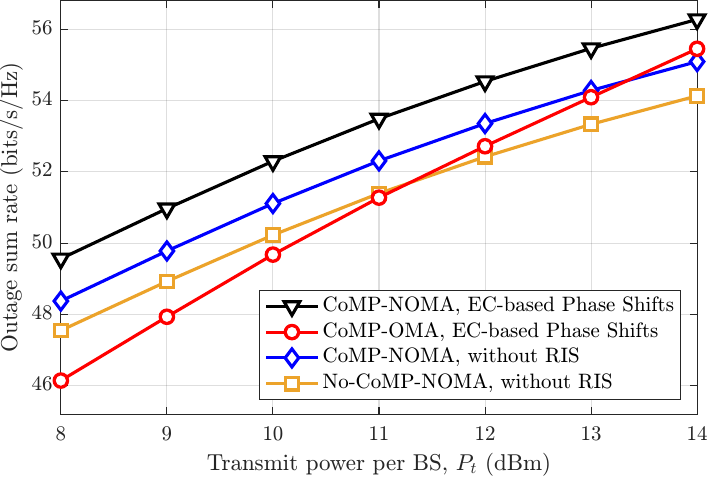}
    \vspace*{0.5em}
    \caption{Outage sum rate vs. transmit power $P_t$ with $J = 4$ cooperative BSs and $K = 70$ elements.}
    \label{fig:eff_osum}
  \end{figure}

  \begin{figure}[h!]
    \centering
    \includegraphics[width=0.65\columnwidth]{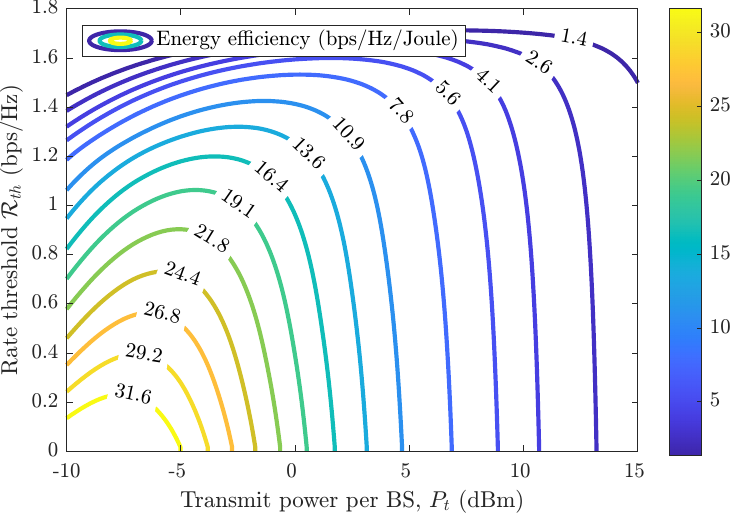}
    \vspace*{0.5em}
    \caption{Energy efficiency contour plot for varying transmit power $P_t$ and rate threshold $\mathcal{R}_{th}$ with $J = 4$ cooperative BSs and $K = 70$ elements.}
    \label{fig:eff_exhaustive}
  \end{figure}

  \section{Impact of RIS Elements on Energy Efficiency}
  The impact of the number of RIS elements $K$ on the energy efficiency is depicted in Figure~\ref{fig:eff_bar}. A clear trend emerges across all configurations: energy efficiency improves as the number of RIS elements increases, reaching a peak at $K = 90$. Beyond this point, the gains in outage sum rate are counterbalanced by the increased power consumption associated with additional elements, resulting in a decline in energy efficiency. Notably, the EC configuration consistently outperforms the EO configuration across all values of $K$. Furthermore, the No-CoMP scenario exhibits the lowest energy efficiency for all values of $K$, underscoring the significance of CoMP in enhancing overall network performance.

  \begin{figure}[t!]
    \centering
    \includegraphics[width=0.65\columnwidth]{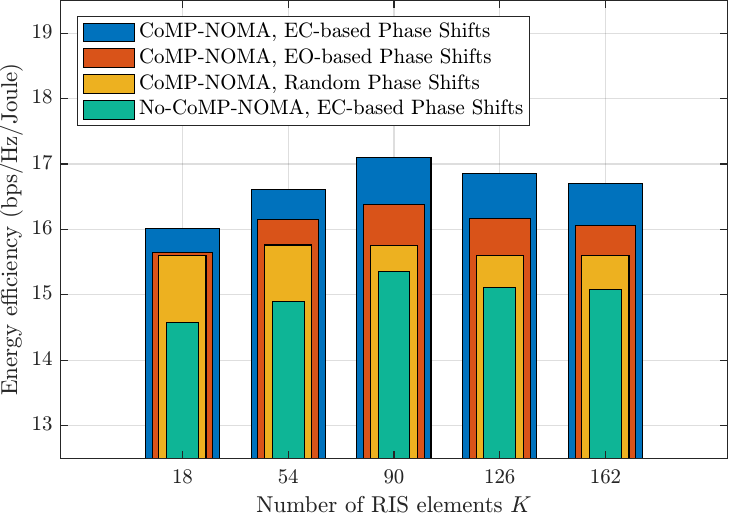}
    \vspace*{0.5em}
    \caption{Energy efficiency vs. number of RIS elements $K$ at $P_t = 0$ dBm and $J = 4$ cooperative BSs.}
    \label{fig:eff_bar}
  \end{figure}

  \begin{figure}[h!]
    \centering
    \includegraphics[width=0.65\columnwidth]{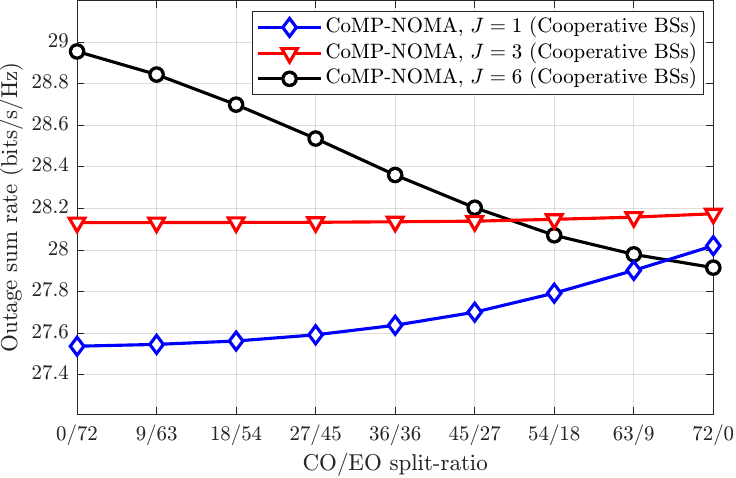}
    \vspace*{0.5em}
    \caption{Outage sum rate vs. CO/EO split-ratio for different cooperative BSs $J$ with $P_t = 0$ dBm with $K = 72$ elements.} \label{fig:eff_split}
  \end{figure}

  \section{Outage Sum Rate}
  The analysis of outage sum rate concerning transmit power $P_t$ is illustrated in Figure~\ref{fig:eff_osum}. To provide a comparative context, orthogonal multiple-access (OMA) is included in the evaluation. As anticipated, the outage sum rate demonstrates an increasing trend with a rise in transmit power across all considered configurations. Notably, the implementation of CoMP-NOMA with EC-based RIS consistently surpasses all other scenarios, including OMA, throughout the entire range of $P_t$ values. An interesting observation is that the CoMP-OMA configuration exhibits steeper increases in outage sum rate per unit of power compared to the CoMP-NOMA configurations. This characteristic can be attributed to the more pronounced rate saturation effect in NOMA networks, wherein the achievable rate is constrained by the SIC capabilities of the users.

  \section{Rate Threshold and Energy Efficiency}
  Figure~\ref{fig:eff_exhaustive} presents a contour plot illustrating the relationship between energy efficiency, transmit power $P_t$, and joint rate threshold $\mathcal{R}_{th}$, where we set $\mathcal{R}^i_{th} = \mathcal{R}^e_{th} = \mathcal{R}_{th}$ for simplicity. Although the highest levels of energy efficiency are found in regions with low $P_t$ and $\mathcal{R}_{th}$, these operating points are often impractical due to minimum user rate requirements. As we move diagonally across the plot, a clear trend emerges: energy efficiency decreases as the rate threshold increases. This inverse relationship highlights the inherent trade-off between achieving higher data rates and maintaining energy efficiency. To meet the demand for higher $\mathcal{R}_{th}$, the network necessitates higher $P_t$, leading to increased power consumption and subsequently, reduced energy efficiency.

  \section{Impact of PBF Design on Sum Rate}
  Lastly, to summarize the trade-off between PBF designs, we analyze the outage sum rate as a function of the ratio of Cancellation-only (CO) and EO elements in Figure~\ref{fig:eff_split}. As expected, increasing CO elements led to a decrease in outage sum rate for the fully cooperative scenario, $J=I$. However, with half of the BSs cooperating, the outage sum rate remained relatively stable regardless of the CO/EO ratio, demonstrating a balanced contribution from both designs. Interestingly, the highest gains were observed when $J=1$, with all elements employing the CO scheme. This emphasizes the critical role of interference cancellation in such scenarios. Moreover, these findings highlight the importance of optimizing network parameters through a robust optimization framework, a direction we intend to explore in the future.

\chapter{Deep Reinforcement Learning for Intelligent NOMA Networks}
\label{chap:drl}
While the previous chapter explored the potential of static RIS configurations in CoMP-NOMA networks, this chapter delves into the dynamic adaptation of these networks through the integration of Unmanned Aerial Vehicles (UAVs) and Aerial RIS (ARIS). This combination of technologies presents a unique opportunity to further enhance network performance by dynamically adjusting the propagation environment, enabling flexible coverage, and optimizing resource allocation.
Existing research has demonstrated the potential of incorporating UAVs within RIS, CoMP, and NOMA frameworks, showcasing notable improvements in various performance metrics as evidenced in~\cite{zhao2022ris, budhiraja2022energy}.

However, many of these studies rely on static RIS deployments, limiting the adaptability and flexibility of the network. Although some works have investigated ARIS-assisted CoMP-NOMA networks and optimized UAV trajectory and RIS phase shifts for sum rate maximization~\cite{lv2023uav}, the optimization methods employed, such as double-layer alternating optimization, often face scalability issues due to their complexity and potential convergence challenges.

To overcome these limitations, this chapter introduces a novel approach based on Deep Reinforcement Learning (DRL) for optimizing ARIS-assisted CoMP-NOMA networks. Our DRL framework jointly optimizes UAV trajectory, RIS phase shifts, and NOMA power control, aiming to maximize the network sum rate while adhering to user Quality of Service (QoS) constraints. Through extensive simulations, we evaluate the effectiveness of our proposed approach, assess the convergence behavior of MO-PPO, and highlight the significant benefits of integrating CoMP-NOMA and RIS within UAV-assisted networks.

\LARGE{\textbf{Towards Optimization}}
\normalsize

\section{System Model}
As shown in Figure~\ref{fig:rl_system}, we consider a multi-cell CoMP-NOMA network assisted by a UAV-mounted RIS in a downlink transmission scenario. The network consists of $I$ cells, each modeled as a circular disk of radius $R$ with a single-antenna BS at its center, denoted as BS$_i$, where $i \in \mathcal{I} \triangleq \{1, 2, \ldots, I\}$. Each BS$_i$ invokes two-user downlink NOMA to serve its respective cell-center and edge user, each also equipped with a single-antenna. The cell-center users are defined as users that lie within the disk of their associated cell and are denoted as U$_{c_i}$, $\forall i$ and $c_i \in \mathcal{C}^i \triangleq \{1, 2, \ldots, C_i\}$, where $C_i$ is the number of cell-center users in cell $i$. Conversely, the edge users are defined as users that do not lie within any cell and are denoted as U$_{f}$, $\forall f$ and $f \in \mathcal{F} \triangleq \{1, 2, \ldots, F\}$, where $F$ is the number of edge users in the network. Furthermore, let $\mathcal{U} \triangleq \bigcup_{i \in \mathcal{I}} \mathcal{C}^i \cup \mathcal{F}$ be the set of all the users in the network. Without loss of generality and for ease of exposition, we assume $I = 2$, and $C_i = F = 1$, $\forall i$.

\begin{figure}[h!]
  \centering
  \includegraphics[width=0.65\columnwidth]{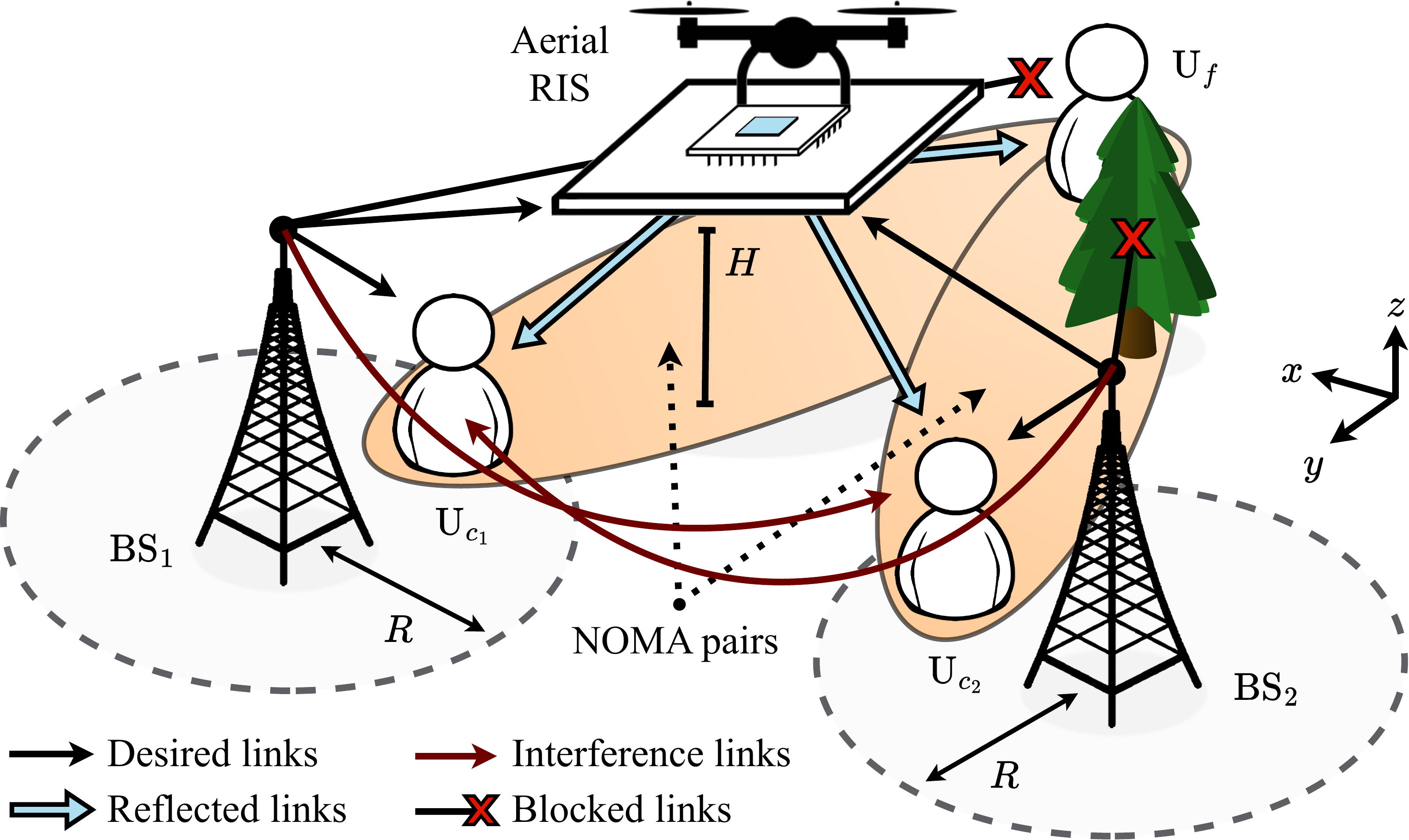}
  \vspace*{0.5em}
  \caption{Aerial RIS-assisted coordinated NOMA cluster.}
  \label{fig:rl_system}
\end{figure}

For coordinated operation, the BSs are assumed to be interconnected via a high-speed backhaul network to a central processing unit (CPU). Moreover, to improve the signal quality for edge users, an ARIS, denoted as $R$, is deployed at a fixed altitude $H$ over area $A$ to create reflection links between the BSs and the users, and is equipped with $K$ passive elements. For tractability, we discretize the entire system operation into time slots of equal length $\tau$, where each time slot is indexed by $t \in \mathcal{T} \triangleq \{1, 2, \ldots, T\}$, such that $T$ is the total flight time of the UAV. Furthermore, we assume the presence of $O$ obstacles in the network, denoted as $\mathcal{O} \triangleq \{1, 2, \ldots, O\}$, where each obstacle O$_o$, $o \in \mathcal{O}$ has its own \textit{forbidden zone} represented as a circular disk of radius $d_{\text{min}}$, centered at the obstacle's location.

Before proceeding with the channel and signal model, we define the positions of the various entities in the network. Specifically, $\forall i \in \mathcal{I}$, $u \in \mathcal{U}$, and $o \in \mathcal{O}$, the positions of BS$_i$, U$_u$, and O$_o$ are represented by $\mathbf{p}_{i}= (x_i, y_i, H_{\text{B}})$, $\mathbf{p}_{u}= (x_u, y_u, 0)$, and $\mathbf{p}_{o}= (x_o, y_o, H_{\text{O}})$, respectively, where $H_{\text{B}}$ and $H_{\text{O}}$ are the heights of the BSs and obstacles, respectively. Moreover, the position of $R$ at time slot $t$ is denoted as $\mathbf{p}_{\text{R}}[t] = (x_{\text{R}}[t], y_{\text{R}}[t], H_{\text{R}})$. In this work, we assume that the users are stationary, and the UAV is capable of adjusting its horizontal position in the $xy$-plane, while maintaining a fixed altitude $H_{\text{R}}$.

\subsection{Channel Model \& RIS Configuration}
In our analysis, both the large-scale path loss and small-scale fading effects are considered. Similar to~\cite{zhao2022ris}, we assume the presence of numerous scatterers in the environment, and thus the direct links between BS$_i$ and U$_u$, denoted as $h_{i,u}$, are modeled as Rayleigh fading channels. Mathematically, the channel $h_{i,u}$ at time slot $t$ is given by
\begin{equation}
  h_{i, u}[t] = \sqrt{\frac{\rho_o}{PL(d_{i, u})}} v_{i, u}[t],
\end{equation}
where $\rho_o$ is the reference path loss at $1$ m, $PL(d_{i, u}) = (d_{i, u})^{-\alpha_{i, u}}$ is the large-scale path loss, such that $\alpha_{i, u}$ is the path loss exponent, and $d_{i, u}= \norm{\mathbf{p}_i - \mathbf{p}_u}$ is the distance between BS$_i$ and U$_u$ and $\norm{\cdot}$ denotes the Euclidean norm. Moreover, $v_{i, u}[t] \in \mathbb{C}^{1\times 1}$ is the small-scale Rayleigh fading coefficient with zero mean and unit variance, and is assumed to be independent and identically distributed (i.i.d) across different time slots and users. In this work, as a special case, we assume that the direct link between BS$_i$ and U$_f$ is blocked due to the presence of obstacles, thus $h_{i, f}[t] = 0$, $\forall i, f$.

Contrary to the direct links, the reflection links between BS$_i$ and $R$ are modeled as Rician fading channels, denoted as $\mathbf{h}_{i, \text{R}}[t]$, due to the presence of a dominant line-of-sight (LoS) component. At time slot $t$, the channel $\mathbf{h}_{i, \text{R}}[t]$ is given by
\begin{equation}
  \mathbf{h}_{i, \text{R}}[t] = \sqrt{\frac{\rho_o}{PL(d_{i, \text{R}}[t])}} \left(\sqrt{\frac{\kappa}{1 + \kappa}} \mathbf{g}^{\text{LoS}}_{i, \text{R}}[t] + \sqrt{\frac{1}{1 + \kappa}} \mathbf{g}^{\text{NLoS}}_{i, \text{R}}[t]\right),
\end{equation}
where $\kappa$ is the Rician factor representing the ratio of the power of the LoS component to the power of the scattered components, and $d_{i, \text{R}}[t] = \norm{\mathbf{p}_i - \mathbf{p}_{\text{R}}[t]}$ is the distance between BS$_i$ and $R$. Moreover, the deterministic LoS represented, i.e., $\mathbf{g}_{i, \text{R}}^{\text{LoS}}[t] \in \mathbb{C}^{K \times 1}$, is given by
\begin{equation*}
  \mathbf{g}^{\text{LoS}}_{i, \text{R}} = \left[1, \ldots, e^{j(k-1)\pi\sin(\omega_i)}, \ldots, e^{j(K-1)\pi\sin(\omega_i)}\right]^T,
\end{equation*}
where $k \in \mathcal{K} \triangleq \{1, 2, \ldots, K\}$ indexes elements of $R$ and $\omega_i$ is the angle of arrival (AoA) whereas $\mathbf{g}^{\text{NLoS}}_{i, \text{R}_i} \in \mathbb{C}^{K\times 1}$ is the NLoS component following Rayleigh fading as previously described. Similarly, the channel between $R$ and U$_u$, denoted as $h_{\text{R}, u}$, can also be modeled as a Rician fading channel.

For the RIS configuration, we assume that the phase shift of the $k$-th element can be set independently of other elements and that both the UAV trajectory and the phase response are controlled by the CPU. Furthermore, the phase shift (PS) matrix at time slot $t$ is expressed as
\begin{equation}
  \mathbf{\Theta}[t] = \text{diag}\left(a_1 e^{j\theta_1[t]}, a_2 e^{j\theta_2[t]}, \ldots, a_k e^{j\theta_K[t]}\right),
\end{equation}
where $a_k \in (0, 1]$ is the amplitude coefficient and $\theta_k[t] \in [-pi, pi)$ is the phase shift of the $k$-th element. In this work, we assume an ideal RIS with perfect phase control and all reflection elements having a unit amplitude, i.e., $a_k = 1, \forall k$. Furthermore, we assume the availability of perfect channel state information (CSI) at the CPU. While this is a challenging assumption in practice, recent advancements in channel estimation techniques for RIS-assisted wireless networks have demonstrated the feasibility of achieving accurate CSI with reasonable overhead~\cite{zheng2022survey, zhou2022channel, wei2021channel}.

\subsection{Signal Model}
In accordance with the NOMA principle, each BS$_i$ serves two users, U$_{c_i}$ and U$_f$, simultaneously, by superimposing their signals. Let $x_{i, c_i}[t]$ and $x_{i, f}[t]$ be the desired signals intended for U$_{c_i}$ and U$_f$ at time slot $t$, then the transmitted signal from BS$_i$ can be expressed as $x_i[t] = \sqrt{(1 - \lambda_i)P_i}x_{i, c_i}[t] + \sqrt{\lambda_i P_i}x_{i, f}[t]$, where $P_i$ is the transmit power of BS$_i$ and $\lambda_i$ is the power allocation factor assigned to U$_f$, such that $\lambda_i \in (0.5, 1)$ to ensure successful decoding at U$_{c_i}$~\cite{elhattab2022ris, obeed2020user}.

The received signal at U$_{f}$ can be expressed as
\begin{equation}
  y_{f}[t] = H_{i, f}[t] x_i[t] + H_{i', f}[t] x_{i'}[t] + n_o[t]
\end{equation}
where $i' \in \mathcal{I} \setminus \{i\}$, $n_o[t] \sim \mathcal{CN}(0, \sigma^2)$ is the additive white Gaussian noise (AWGN), and $H_{i, f}[t] = \mathbf{h}^T_{\text{R}, f}[t] \mathbf{\Theta}[t] \mathbf{h}_{i, \text{R}}[t]$ represents the effective channels between BS$_i$ and U$_f$ through $R$, respectively. To minimize synchronization overhead, we employ non-coherent JT-CoMP, thus, the signal-to-interference-plus-noise ratio (SINR) is given by
\begin{equation}
  \gamma_f[t] = \frac{\lambda_i \abs{H_{i, f}[t]}^2 + \lambda_{i'} \abs{H_{i', f}[t]}^2}{(1 - \lambda_i) \abs{H_{i, f}[t]}^2 + (1 - \lambda_{i'}) \abs{H_{i', f}[t]}^2 + \frac{1}{\rho}},
\end{equation}
where $\rho = P_t/\sigma^2$ is the transmit SNR and $P_t = P_i$, $\forall i$ is the transmit power of each BS.

On the other hand, the received signal at U$_{c_i}$ can be expressed as
\begin{equation}
  y_{c_i}[t] = H_{i, c_i}[t] x_i[t] + h_{i', c_i}[t] x_{i'}[t] + n_o[t],
\end{equation}
where $H_{i, c_i}[t] = h_{i, c_i}[t] + \mathbf{h}^T_{\text{R}, c_i}[t] \mathbf{\Theta}[t] \mathbf{h}_{i, \text{R}}[t]$ represents the effective channels between BS$_i$ and U$_{c_i}$ through $R$, respectively. Also, the term $h_{i', c_i}[t] x_{i'}[t]$ represents the ICI caused by the transmission of BS$_{i'}$ at U$_{c_i}$. Based on the SIC principle, U$_{c_i}$ first decodes $x_{i, f}[t]$ and then cancels it from $y_{c_i}[t]$ to decode $x_{i, c_i}[t]$. The SINR at U$_{c_i}$ for decoding $x_{i, f}[t]$ is given by
\begin{equation}
  \gamma_{c_i \rightarrow f}[t] = \frac{\lambda_i \abs{H_{i, c_i}[t]}^2}{(1 - \lambda_i) \abs{H_{i, c_i}[t]}^2 + \abs{h_{i', c_i}[t]}^2 + \frac{1}{\rho}},
\end{equation}
whereas the SINR at U$_{c_i}$ for decoding $x_{i, c_i}[t]$ is
\begin{equation}
  \gamma_{c_i}[t] = \frac{(1 - \lambda_i) \abs{H_{i, c_i}[t]}^2}{\abs{h_{i', c_i}[t]}^2 + \frac{1}{\rho}}.
\end{equation}
Finally, the achievable sum rate of the network at time slot $t$ can be expressed as
\begin{equation}
  R_{\text{sum}}[t] = \sum_{i \in \mathcal{I}} R_{c_i}[t] + \sum_{f \in \mathcal{F}} R_f[t].
\end{equation}
where $R_{c_i}[t]=\log_2(1 + \gamma_{c_i}[t])$ and $R_f[t]=\log_2(1 + \gamma_f[t])$ are the achievable rates of U$_{c_i}$ and U$_f$, respectively.

\section{Problem Formulation}
In this work, our primary objective is to maximize the sum rate achieved over $T$ time slots. To achieve this goal, we jointly optimize three key control variables: the UAV trajectory denoted as $\mathbf{P} \triangleq \{\mathbf{p}_{\text{R}}[t], \forall t\}$, the RIS phase shifts represented by $\mathbf{\Theta} \triangleq \{\mathbf{\Theta}[t], \forall t\}$, and the power allocation factors denoted as $\mathbf{\Lambda} \triangleq \{\lambda_i, \forall i\}$.

The optimization problem can be mathematically formulated as
\begin{maxi!}|s|
{\mathbf{P}, \mathbf{\Theta}, \mathbf{\Lambda}}{\sum_{t \in \mathcal{T}} R_{\text{sum}}[t]}{\label{eq:problem}}{}
\addConstraint{x_{\text{R}}[t], y_{\text{R}}[t] \in A,~\forall t \in \mathcal{T} \label{eq:loc_csrt}}
\addConstraint{\norm{\mathbf{p}_{\text{R}}[t] - \mathbf{p}_{o}} \geq d_{\text{min}},~\forall o \in \mathcal{O},~t \in \mathcal{T} \label{eq:obs_csrt}}
\addConstraint{\theta_k[t] \in [-\pi, \pi),~\forall k \in \mathcal{K},~t \in \mathcal{T} \label{eq:phase_csrt}}
\addConstraint{R_{c_i}[t] \geq R_{c_i}^{\text{min}},~\forall i \in \mathcal{I},~t \in \mathcal{T} \label{eq:ratec_csrt}}
\addConstraint{R_{f}[t] \geq R_{f}^{\text{min}},~\forall f \in \mathcal{F},~t \in \mathcal{T} \label{eq:ratef_csrt}}
\addConstraint{\lambda_i \in (0.5, 1),~\forall i \in \mathcal{I} \label{eq:power_csrt}},
\end{maxi!}
where constraint~\eqref{eq:loc_csrt} restricts the UAV trajectory to lie within $A$, and constraint~\eqref{eq:obs_csrt} enforces a minimum safety distance between the UAV and any obstacles present, thus guaranteeing the UAV's safety. Constraint~\eqref{eq:phase_csrt} limits the phase shifts applied by the RIS elements. To meet the quality of service (QoS) requirements, constraints~\eqref{eq:ratec_csrt} and~\eqref{eq:ratef_csrt} impose minimum rate thresholds, denoted by $\mathcal{R}_{c_i}^{\text{min}}$ and $\mathcal{R}_{f}^{\text{min}}$, for U$_{c_i}$ and U$_f$, respectively. Lastly, constraint~\eqref{eq:power_csrt} defines the permissible range for power allocation factors, ensuring successful SIC. The optimization problem in~\eqref{eq:problem} is non-convex due to the coupled variables $\{\mathbf{P}, \mathbf{\Theta}, \mathbf{\Lambda}\}$. To address this, we propose a DRL-based solution in the next section.

\section[DRL-based Proposed Solution]{Deep Reinforcement Learning-based Proposed Solution}

\subsection{MDP Formulation}
Before proceeding with the DRL-based solution, we model it as a single-agent Markov Decision Process (MDP) with discrete time steps. This MDP is represented by the tuple $\langle \mathcal{S}, \mathcal{A}, \mathcal{P}, \mathcal{R}, \gamma \rangle$, where $\mathcal{S}$ denotes the set of possible environment states, $\mathcal{A}$ represents the action space, $\mathcal{P}$ defines the state transition probabilities, $\mathcal{R}$ is the reward function guiding the agent's learning, and $\gamma$ is the discount factor that determines the importance of future rewards. At each time slot $t$, the agent observes the current state $s_t$, selects an action $a_t$ based on its policy, transitions to a new state $s_{t+1}$, and receives a reward $\mathcal{R}(s_t, a_t)$. We define $\mathcal{S}$, $\mathcal{A}$, and $\mathcal{R}$ as follows
\begin{enumerate}[wide=\parindent]
  \item \textit{State Space $\mathcal{S}$:} The environment state at time slot $t$ consists of the UAV's current position $\mathbf{p}_{\text{R}}[t]$, the distance from the UAV to the center of obstacles $\mathbf{d}_{\text{R}}[t] = \{\norm{\mathbf{p}_{\text{R}}[t] - \mathbf{p}_{o}}, \forall o \in \mathcal{O}\}$, the power allocation factors $\mathbf{\Lambda}$, and the achievable rates $\mathbf{R}[t] = \{R_{c_i}[t], R_f[t], \forall i, f\}$. Thus, the state space can be expressed as
        \begin{equation}
          s_t = \{\mathbf{p}_{\text{R}}[t], \mathbf{d}_{\text{R}}[t], \mathbf{\Lambda}, \mathbf{R}[t]\} \in \mathbb{R}^{\text{dim}_\mathcal{S}}.
        \end{equation}
        where $\text{dim}_\mathcal{S}=2 + O + I + \sum_{i \in \mathcal{I}} C_i + F$ is the dimension of the state space.

  \item \textit{Action Space $\mathcal{A}$:} The action space of the formulated MDP consists of the UAV's movement in the horizontal $xy$-plane, the phase shifts of the RIS elements, and the power allocation factors. Specifically, the action space at time slot $t$ contains the manuevering actions $\mathbf{a}_{\text{R}}[t] \in \{(-1, 0), (1, 0), (0, -1), (0, 1), (0, 0)\}$, representing left, right, down, up, and hover, respectively, the phase shifts $\mathbf{a}_{\Phi}[t] = \{\phi_k[t], \forall k\}$, and the power allocation factors $\mathbf{a}_{\Lambda} = \{\lambda_i, \forall i\}$. Thus, the action space can be expressed as
        \begin{equation}
          \mathbf{a}_t = \{\mathbf{a}_{\text{R}}[t], \mathbf{a}_{\Phi}[t], \mathbf{a}_{\Lambda}\} \in \mathbb{R}^{\text{dim}_\mathcal{A}}.
        \end{equation}
        where $\text{dim}_\mathcal{A}=2 + K + I$ is the dimension of the action space.

  \item \textit{Reward Function $\mathcal{R}$:} The reward function plays a crucial role in shaping the learning behavior of the RL agent. Our design encourages maximizing the sum rate while ensuring UAV safety and meeting QoS requirements by penalizing constraint violations. The reward function is defined as
        \begin{equation}
          \label{eq:reward}
          \mathcal{R}(s_t, a_t) = R_{\text{sum}}[t] \left(1 - \frac{\sum_{u \in \mathcal{U}} \zeta_u[t]}{|\mathcal{U}|}\right) - \xi_{\text{R}}[t] K_{\text{viol}},
        \end{equation}
        where $K_{\text{viol}}$ is the penalty factor for constraint violation, and $\zeta_u[t]=\mathbb{I}\{R_u[t] \leq R_u^{\text{min}}\}$ is the indicator function for the QoS constraints, i.e., $\zeta_u[t]=1$ if QoS constraints are violated, and $0$ otherwise. Similarly, $\xi_{\text{R}}[t]=\mathbb{I}\{x_{\text{R}}[t], y_{\text{R}}[t] \notin A \land \norm{\mathbf{p}_{\text{R}}[t] - \mathbf{p}_{o}} < d_{\text{min}}, \forall o \in \mathcal{O}\}$ is the indicator function for UAV's safety constraints.
\end{enumerate}

\subsection{MO-PPO Algorithm}
\begin{algorithm}[t]
  \DontPrintSemicolon
  \caption{MO-PPO Algorithm}\label{alg:moppo}
  Initialize the policy parameters $\theta_d$ and $\theta_c$\;
  \For{episode $= 1, 2, \ldots, N$}{
  Receive initial state $s_0$\;
  \For{time step $t = 0, 1, \ldots, T$}{
    Generate discrete action $a_{\text{R}}$ using $\pi_{\theta_d}(a_t|s_t)$\;
    Generate continuous actions $\mathbf{a}_{\Phi}$ and $\mathbf{a}_{\Lambda}$\ using $\pi_{\theta_c}(a_t|s_t)$\;
    Execute actions $\mathbf{a}_t = \{a_{\text{R}}, \mathbf{a}_{\Phi}, \mathbf{a}_{\Lambda}\}$\;
    \If{UAV violates~\eqref{eq:loc_csrt} or~\eqref{eq:obs_csrt}}{
      Set $\xi_{\text{R}}[t] = 1$, cancel the UAV's movement, and update the state $s_{t+1}$\;
    }
    Observe reward $\mathcal{R}$ as~\eqref{eq:reward} and next state $s_{t+1}$\;
    Collect a set of partial trajectories $\mathcal{D}$ with $\hat{T}$ transitions\;
    Compute the advantage estimate $\hat{A}_t$ as~\eqref{eq:adv}\;
  }
  \For{epoch $= 1, 2, \ldots, E$}{
  Sample a mini-batch of transitions $B$ from $\mathcal{D}$\;
  Compute the clipped surrogate objectives $L_d^{\text{CLIP}}(\theta_d)$ and $L_c^{\text{CLIP}}(\theta_c)$ as~\eqref{eq:clip}\;
  Optimize overall objective and update the policy parameters $\theta_d$ and $\theta_c$ using Adam~\cite{kingma2014adam}\;
  }

  Synchronize the sampling policies as\;
  \centerline{$\theta_d^{\text{old}} \leftarrow \theta_d$ and $\theta_c^{\text{old}} \leftarrow \theta_c$}
  Clear the collected trajectories $\mathcal{D}$\;
  }
\end{algorithm}

In this work, the considered action space is a hybrid continous-discrete space, which poses a challenge for traditional RL algorithms. While discretization of continuous actions is a possibility, it can lead to a large action space, significantly increasing computational complexity and potentially hindering performance. To address this challenge, we propose employing a multi-output Proximal Policy Optimization (MO-PPO) algorithm. MO-PPO extends the standard PPO~\cite{schulman2017proximal} framework by employing two parallel actor networks, each responsible for generating the discrete action $\mathbf{a}_{\text{R}}$ and the continuous actions $\mathbf{a}_{\Phi}$ and $\mathbf{a}_{\Lambda}$, respectively. The actor networks share the first few layers, allowing for the extraction of common features and encoding the state information. Furthermore, a single critic network is employed to estimate the value function $V(s_t)$, which is used to compute a variance-reduced advantage function estimate $\hat{A}_t$ for policy optimization. Following the implementation details used in~\cite{mnih2016asynchronous}, the policy is executed for $\hat{T}$ time steps, and $\hat{A}_t$ is computed as
\begin{equation}
  \label{eq:adv}
  \hat{A}_t = \sum_{k=0}^{\hat{T}-1} \gamma^k r_{t+k} + \gamma^{\hat{T}} V(s_{t+\hat{T}}) - V(s_t),
\end{equation}
where $\hat{T}$ is much smaller than the length of the episode $T$.

To generate the stochastic policy $\pi_{\theta_d}(a_t|s_t)$ for the discrete actions, the corresponding actor network outputs $|\mathbf{a}_{\text{R}}|$ logits, which are then passed through a $\mathrm{softmax}$ function to obtain a probability distribution over the available discrete actions. Conversely, the continuous actor network generates the continuous actions $\mathbf{a}_{\Phi}$ and $\mathbf{a}_{\Lambda}$ by sampling from Gaussian distributions parameterized by the mean and standard deviation outputs of the network, as dictated by the stochastic policy $\pi_{\theta_c}(a_t|s_t)$. Both $\pi_{\theta_d}(a_t|s_t)$ and $\pi_{\theta_c}(a_t|s_t)$ are optimized independently using their respective clipped surrogate objective functions. For the discrete actions, the objective function is given by
\begin{equation}
  \label{eq:clip}
  L_d^{\text{CLIP}}(\theta_d) = \hat{\mathbb{E}}_t \left[\min(r_t^d(\theta_d) \hat{A}_t, \aleph(r_t^d,\theta_d, \epsilon) \hat{A}_t\right],
\end{equation}
where $\aleph(r_t^d,\theta_d, \epsilon) = \text{clip}(r_t^d(\theta_d), 1 - \epsilon, 1 + \epsilon)$, $r_t^d(\theta_d) = \pi_{\theta_d}(a_t|s_t)/\pi_{\theta_d}^{\text{old}}(a_t|s_t)$ is the importance sampling ratio, and $\epsilon$ is the clipping parameter. The objective function for the continuous actions can be expressed in a similar manner but is left out for brevity.

It is important to note that while both policies collaborate within the environment, their optimization objectives remain decoupled, i.e., $\pi_{\theta_d}(a_t|s_t)$ and $\pi_{\theta_c}(a_t|s_t)$ are treated as independent distributions during policy optimization, rather than a joint distribution encompassing both action spaces. The MO-PPO algorithm is summarized in Algorithm~\ref{alg:moppo}.

\subsection{Complexity and Convergence Analysis}
The complexity of DRL algorithms is commonly measured in terms of the number of multiplications per iteration, which is a function of the number of parameters in the policy and value networks. For MO-PPO, the overall complexity can be expressed as $\mathcal{O}[\sum_{q=1}^{Q_s} n_q \cdot n_{q-1} + \sum_{q=1}^{Q_d} n_q \cdot n_{q-1} + \sum_{q=1}^{Q_c} n_q \cdot n_{q-1}]$, where $Q_s$, $Q_d$, and $Q_c$ are the number of layers in the shared, discrete, and continuous actor networks, respectively, and $n_q$ and $n_{q-1}$ are the number of neurons in the $q$-th and $(q-1)$-th layers, respectively. In this work, we consider the same number of neurons in each hidden layer, i.e., $n_q = n_{q-1} = n$, $\forall q$, and the number of neurons in the output layer is equal to the dimension of the action space. Thus, the overall complexity of MO-PPO is $\mathcal{O}[n^2(Q_s + Q_d + Q_c)]$.

Similar to other DRL algorithms, the convergence of MO-PPO is mathematically difficult to analyze~\cite{challita2019interference} since neural networks are highly dependent on the choice of hyperparameters. However, the convergence of MO-PPO can be empirically verified by monitoring the agent's performance over multiple episodes and ensuring that the reward function converges to a stable value. Moreover, the convergence of MO-PPO can be accelerated by tuning the learning rate, clipping parameter, and penalty factor for constraint violation.

\LARGE{\textbf{Numerical Results}}
\normalsize

\section{Simulation Setup}
To evaluate the efficacy of the proposed MO-PPO algorithm, we construct a simulated urban environment spanning an area of $150 \times 150$ m$^2$ with $\mathcal{I} = 2$ BSs, $\mathcal{U} = 3$ users, and $\mathcal{O} = 2$ obstacles. The initial position of the UAV is set to $(0, 35, 50)$ m, while BS$_1$ and BS$_2$ are located at $(-35, -35, 25)$ m and $(35, 35, 25)$ m, respectively. All remaining entities are randomly placed within the environment.

Both BSs are assumed to transmit at an identical power level, i.e., $P_1 = P_2 = P_t$. Furthermore, the network operates at a carrier frequency of $f_c = 2.4$ GHz, utilizing a bandwidth of $BW = 10$ MHz and the noise power is set to $\sigma^2 = -174 + 10\log_{10}(BW)$ dBm. To model the signal propagation characteristics, we employ path loss exponents of $\alpha_{i, u} = 3$, $\alpha_{i, \text{R}} = \alpha_{\text{R}, u} = 2.2$, and $\alpha_{i', u} = 3.5$, for direct, reflection, and interference links, respectively. Table~\ref{tab:rl_params} summarizes the remaining simulation parameters.

\begin{table}[h!]
  \centering
  \caption{Simulation Parameters}
  \label{tab:rl_params}
  \resizebox{0.85\columnwidth}{!}{%
    \begin{tabular}{|c|c|c|c|}
      \hline
      \textbf{Parameter}                      & \textbf{Value} & \textbf{Parameter}            & \textbf{Value}   \\
      \hline
      \hline
      Reference path loss $\rho_o$            & $-30$ dBm      & Rician factor $\kappa$        & $3$ dB           \\
      Target data rate $R_{f}^{\text{min}}$   & $0.2$ bps/Hz   & Learning rate                 & $2.75\text{e}-4$ \\
      Target data rate $R_{c_i}^{\text{min}}$ & $0.5$ bps/Hz   & Clipping parameter $\epsilon$ & $0.1$            \\
      Penalty constant $K_{\text{viol}}$      & $7$            & Discount factor $\gamma$      & $0.98$           \\
      Minimum distance $d_{\text{min}}$       & $10$ m         & Number of episodes $N$        & $750$            \\
      Time slots per episode $T$              & $250$          & Number of epochs $E$          & $20$             \\
      Number of neurons                       & $64$           & Batch size $B$                & $128$            \\
      \hline
    \end{tabular}%
  }
\end{table}

\begin{figure}[t!]
  \centering
  \includegraphics[width=0.65\columnwidth]{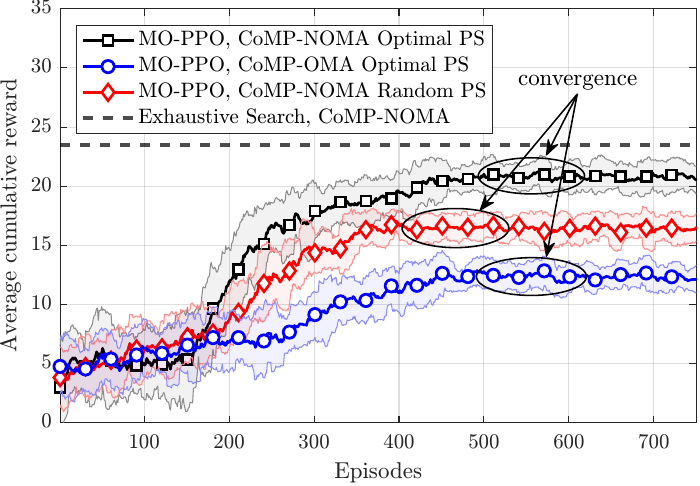}
  \vspace*{0.5em}
  \caption{Average cumulative reward vs. number of training episodes with $P_t = 20$ dBm and $K = 120$ elements.} \label{fig:rl_reward}
\end{figure}

\begin{figure}[h!]
  \centering
  \includegraphics[width=0.65\columnwidth]{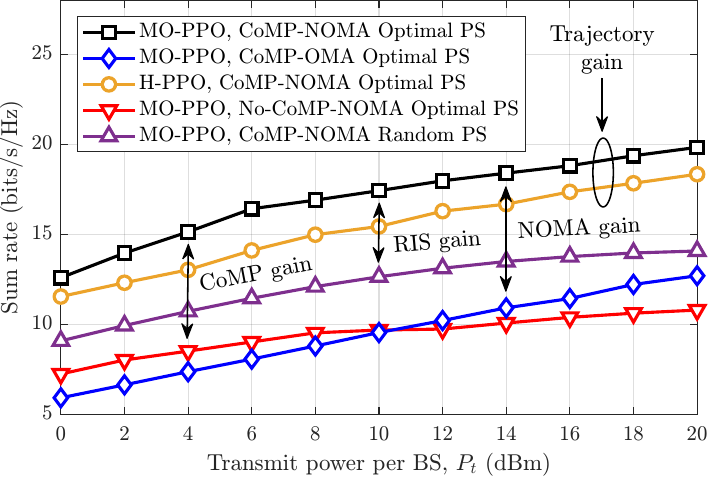}
  \vspace*{0.5em}
  \caption{Sum rate vs. transmit power for different algorithms and configurations with $K = 120$ elements.}
  \label{fig:rl_sum}
\end{figure}

\section{Learning and Convergence}
Figure~\ref{fig:rl_reward} illustrates the average cumulative reward achieved by the MO-PPO algorithm with different network configurations. As shown, the algorithm consistently converges to a stable reward value after approximately 500 episodes, indicating the successful acquisition of an effective policy. Notably, MO-PPO with random PS exhibits a faster convergence rate compared to its counterpart with optimal PS. This observation can be attributed to the increased complexity associated with optimizing the PS within the action space, leading to a slower convergence process. Moreover, the CoMP-NOMA configuration achieves a superior average cumulative reward compared to the CoMP-OMA configuration, underscoring the benefits of NOMA in enhancing overall network performance. A comparison with the optimal solution obtained through exhaustive search reveals that the proposed MO-PPO algorithm achieves near-optimal performance, effectively demonstrating its capability to solve the formulated problem.

\begin{figure}[t!]
  \centering
  \includegraphics[width=0.65\columnwidth]{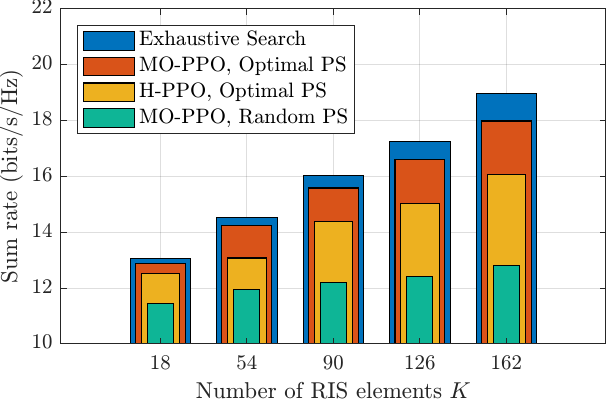}
  \vspace*{0.5em}
  \caption{Impact of the number of RIS elements on the achievable sum rate with $P_t = 10$ dBm.}
  \label{fig:rl_bar}
\end{figure}

\begin{figure}[t!]
  \centering
  \includegraphics[width=0.65\columnwidth]{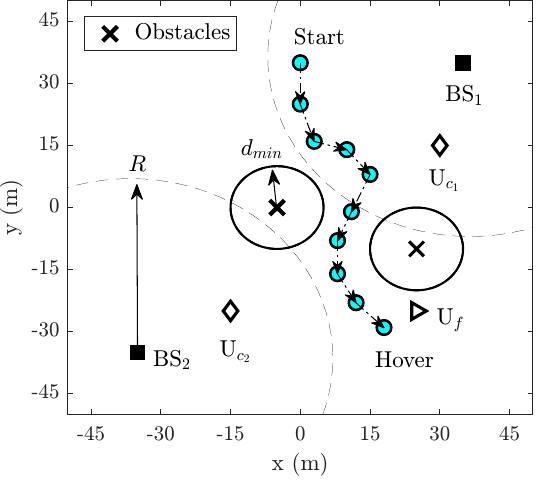}
  \vspace*{0.5em}
  \caption{Top view of the UAV trajectory obtained by the MO-PPO algorithm sampled every 25 time slots and averaged over 10 evaluation episodes}
  \label{fig:rl_traj}
\end{figure}

\section{Algorithmic Performance on Sum Rate}
Next, we investigate the sum rate achieved by the network as a function of the transmit power $P_t$ as shown in Figure~\ref{fig:rl_sum}. As expected, the sum rate exhibits an upward trend with increasing transmit power, emphasizing the crucial role of power control in optimizing network performance. The results clearly showcase the advantages of incorporating CoMP, RIS, and NOMA techniques to enhance spectral efficiency across all power levels. Additionally, we benchmark the hover PPO (H-PPO) algorithm, which maintains a fixed UAV position at the center of the user clusters, against the proposed MO-PPO algorithm, and highlight the improvement in network performance achieved due to the trajectory optimization.

\section{MO-PPO and RIS Elements}
Our investigation extends to analyzing the impact of the number of RIS elements on the network's achievable sum rate, providing further insights into the performance of the MO-PPO algorithm Figure~\ref{fig:rl_bar} illustrates the positive correlation between the sum rate and the number of RIS elements, emphasizing the advantages of utilizing a larger RIS to enhance network performance. However, we observe a subtle, yet noteworthy trend: the difference in sum rate between the exhaustive search baseline and the MO-PPO algorithm, while remaining small, increases with the number of RIS elements. This observation highlights the importance of carefully considering the trade-off between performance and complexity when determining the optimal number of operational RIS elements, a challenge we aim to address in future work.

\section{UAV Trajectory}
Finally, we visualize the UAV trajectory generated by the MO-PPO algorithm in Figure~\ref{fig:rl_traj}. It is observed that the UAV adopts a \textit{cautious approach}, navigating around obstacles while minimizing its distance to U$_f$. Such a trend is commonly observed in DRL algorithms that operate on the principle of exploration-exploitation. The agent learns to strike a balance between exploring the environment and exploiting its current knowledge to maximize the cumulative reward. The generated trajectory further highlights the agent's ability to adapt to the dynamic environment and optimize network performance by effectively leveraging both RIS and NOMA techniques.

\chapter{Conclusion}
Our research explores the potential of utilizing reconfigurable intelligent surfaces (RIS), specifically STAR-RIS and aerial RIS (ARIS), within CoMP-NOMA networks to enhance performance and energy efficiency. Across four papers, we investigate various aspects of network design, optimization, and performance analysis.

\section{Summary}
This thesis has investigated the potential of integrating STAR-RIS, CoMP, and NOMA technologies to address the challenges of future wireless communication systems. The research explored various aspects of these technologies, including performance analysis, optimization strategies, and practical implementation considerations.

\subsection{Performance Analysis}
A novel framework was introduced for integrating STAR-RIS into CoMP-NOMA multi-cell networks, highlighting the unique ability of STAR-RIS to serve multiple cells simultaneously. The impact of STAR-RIS on key performance metrics such as achievable rates and outage probability was analyzed, demonstrating its superiority over conventional systems in improving coverage and performance, particularly for cell-edge users. Additionally, the thesis developed a tractable analytical framework to evaluate the performance of STAR-RIS-assisted CoMP-NOMA networks under Nakagami-m fading channels, deriving closed-form expressions for ergodic rate and outage probability for each user.

\subsection{Optimization Strategies}
The research investigated various optimization strategies for maximizing network performance. This included exploring the impact of CoMP cooperation, STAR-RIS element allocation to base stations based on channel conditions, and amplitude adjustments for transmission and reflection on achieving optimal network sum-rate. Furthermore, the thesis explored energy-efficient design approaches for CoMP-NOMA networks incorporating RIS, proposing different RIS configurations and optimization algorithms for maximizing energy efficiency while maintaining desired performance levels.

\subsection{Deep Reinforcement Learning}
The application of Deep Reinforcement Learning (DRL) techniques was explored for joint optimization of UAV trajectory, RIS phase shifts, and NOMA power control in aerial RIS-assisted CoMP-NOMA networks. The proposed DRL framework demonstrated effectiveness in maximizing network sum rate while satisfying user QoS constraints, highlighting the potential of combining CoMP-NOMA and RIS in UAV-assisted networks for improved spectral efficiency and coverage.

\section{Limitations and Future Work}
While this thesis has made significant contributions to the understanding of RIS-assisted CoMP-NOMA networks, there are limitations that pave the way for future research:
\begin{itemize}
    \item \textbf{Limited Scope of Analysis:} The analysis primarily focused on achievable rates and outage probability, neglecting other crucial aspects such as energy efficiency, spectral efficiency, and user fairness.
    \item \textbf{Simplified Assumptions:} The research often relied on static scenarios and simplified channel models, which may not fully capture the dynamic nature and complexities of real-world networks.
    \item \textbf{Specific Technology Implementations:} The analysis focused on specific RIS configurations and DRL algorithms, potentially overlooking other promising options for optimizing performance.
\end{itemize}

Future research directions include:
\begin{itemize}
    \item \textbf{Investigating more realistic and dynamic scenarios:} This includes considering user mobility, complex channel models, and dynamic interference environments.
    \item \textbf{Expanding performance metrics:} Future work should analyze a wider range of metrics, including energy efficiency, spectral efficiency, user fairness, and delay, to provide a more comprehensive evaluation of network performance.
    \item \textbf{Exploring distributed optimization algorithms:} Developing distributed optimization algorithms for RIS control is crucial for scalability and robustness in large-scale networks.
    \item \textbf{Addressing practical implementation aspects:} Further research is required to address challenges related to hardware limitations, channel estimation, and standardization of RIS technology.
    \item \textbf{Incorporating machine learning techniques:} Utilizing machine learning for tasks such as channel prediction, resource allocation, and interference management can further enhance the performance and adaptability of RIS-assisted CoMP-NOMA networks.
\end{itemize}

By addressing these limitations and pursuing these future research directions, we can continue to advance the development of RIS-assisted CoMP-NOMA networks and unlock their full potential for shaping the future of wireless communication.


\bibliographystyle{ieeetr}

\addcontentsline{toc}{chapter}{References}

\clearpage

\end{document}